\documentclass[sigconf, authorversion, nonacm, screen]{acmart}

\usepackage{graphicx}
\usepackage{textcomp}
\usepackage{mathtools}
\usepackage{color,soul}
\usepackage{url}
\usepackage{algorithm}
\usepackage{enumitem}
\usepackage{multirow}
\usepackage{subcaption}

\usepackage{ifthen}
\usepackage{xargs}

\newcommandx{\yaHelper}[2][1=\empty]{%
\ifthenelse{\equal{#1}{\empty}}%
  { \ensuremath{ \scriptstyle{ #2 } } } 
  { \raisebox{ #1 }[0pt][0pt]{ \ensuremath{ \scriptstyle{ #2 } } } }  
}

\newcommandx{\yrightarrow}[4][1=\empty, 2=\empty, 4=\empty, usedefault=@]{%
  \ifthenelse{\equal{#2}{\empty}}
  { \xrightarrow{ \protect{ \yaHelper[ #4 ]{ #3 } } } } 
  { \xrightarrow[ \protect{ \yaHelper[ #2 ]{ #1 } } ]{ \protect{ \yaHelper[ #4 ]{ #3 } } } } 
}

\newcommandx{\yleftarrow}[4][1=\empty, 2=\empty, 4=\empty, usedefault=@]{%
  \ifthenelse{\equal{#2}{\empty}}
  { \xleftarrow{ \protect{ \yaHelper[ #4 ]{ #3 } } } } 
  { \xleftarrow[ \protect{ \yaHelper[ #2 ]{ #1 } } ]{ \protect{ \yaHelper[ #4 ]{ #3 } } } } 
}

\newcommandx{\yRightarrow}[4][1=\empty, 2=\empty, 4=\empty, usedefault=@]{%
  \ifthenelse{\equal{#2}{\empty}}
  { \xRightarrow{ \protect{ \yaHelper[ #4 ]{ #3 } } } } 
  { \xRightarrow[ \protect{ \yaHelper[ #2 ]{ #1 } } ]{ \protect{ \yaHelper[ #4 ]{ #3 } } } } 
}

\newcommandx{\yLeftarrow}[4][1=\empty, 2=\empty, 4=\empty, usedefault=@]{%
  \ifthenelse{\equal{#2}{\empty}}
  { \xLeftarrow{ \protect{ \yaHelper[ #4 ]{ #3 } } } } 
  { \xLeftarrow[ \protect{ \yaHelper[ #2 ]{ #1 } } ]{ \protect{ \yaHelper[ #4 ]{ #3 } } } } 
}

\usepackage{tikz}
\usepackage{pgfplots}
\usepackage{grffile}
\pgfplotsset{compat=newest}
\usetikzlibrary{automata}
\usetikzlibrary{arrows, arrows.meta, decorations.markings, patterns}
\usepgfplotslibrary{patchplots, groupplots}
\usetikzlibrary{shapes, positioning, fit, backgrounds}
\usepgfplotslibrary{fillbetween}

\usetikzlibrary{overlay-beamer-styles}
\usetikzlibrary{calc}

\usepackage{calc}

\definecolor{tudcyan}{RGB}{0,166,214}
\definecolor{tudmagenta}{RGB}{109,23,127}
\definecolor{tudpurple}{RGB}{29,28,115}
\definecolor{tudgraygreen}{RGB}{107,134,137}

\colorlet{lighttudcyan}{tudcyan!20}
\colorlet{lighttudmagenta}{tudmagenta!20}

\newlength{\hatchspread}
\newlength{\hatchthickness}
\newlength{\hatchshift}
\newcommand{\hatchcolor}{}
\tikzset{hatchspread/.code={\setlength{\hatchspread}{#1}},
	hatchthickness/.code={\setlength{\hatchthickness}{#1}},
	hatchshift/.code={\setlength{\hatchshift}{#1}},
	hatchcolor/.code={\renewcommand{\hatchcolor}{#1}}}
\tikzset{hatchspread=7pt,
	hatchthickness=0.5pt,
	hatchshift=0pt,
	hatchcolor=black}
\pgfdeclarepatternformonly[\hatchspread,\hatchthickness,\hatchshift,\hatchcolor]
{custom north west lines}
{\pgfqpoint{\dimexpr-2\hatchthickness}{\dimexpr-2\hatchthickness}}
{\pgfqpoint{\dimexpr\hatchspread+2\hatchthickness}{\dimexpr\hatchspread+2\hatchthickness}}
{\pgfqpoint{\dimexpr\hatchspread}{\dimexpr\hatchspread}}
{
	\pgfsetlinewidth{\hatchthickness}
	\pgfpathmoveto{\pgfqpoint{0pt}{\dimexpr\hatchspread+\hatchshift}}
	\pgfpathlineto{\pgfqpoint{\dimexpr\hatchspread+0.15pt+\hatchshift}{-0.15pt}}
	\ifdim \hatchshift > 0pt
	\pgfpathmoveto{\pgfqpoint{0pt}{\hatchshift}}
	\pgfpathlineto{\pgfqpoint{\dimexpr0.15pt+\hatchshift}{-0.15pt}}
	\fi
	\pgfsetstrokecolor{\hatchcolor}
	\pgfusepath{stroke}
}

\def\centerarc[#1](#2)(#3:#4:#5)
{ \draw[#1] ($(#2)+({#5*cos(#3)},{#5*sin(#3)})$) arc (#3:#4:#5); }

\AtBeginDocument{%
	\providecommand\BibTeX{{%
			\normalfont B\kern-0.5em{\scshape i\kern-0.25em b}\kern-0.8em\TeX}}}

\DeclareMathOperator{\Pre}{Pre}  
\DeclareMathOperator{\Post}{Post}

\def\norm[#1]{\left|#1\right|}
\def\shortnorm[#1]{|#1|}
\def\bisim{\ensuremath{\cong}}
\def\MAS{\ensuremath{\mathsf{MAS}}}
\def\equivAS{\ensuremath{\simeq_{AS}}}

\def\trsize{\mathsf{TranSize}}
\def\leAS{\ensuremath{\preceq_{AS}}}
\def\smin{\mathcal{S}_{\mathsf{min}}}
\def\isop{\ensuremath{\cong_{\mathsf{is}}}}
\def\MAS{\ensuremath{\text{MAS}}}
\def\MAS{\ensuremath{R^{\max}}}
\newcommand{\aA}{\mathcal{A}}

\newcommand{\parts}{\mathsf{Part}}
\newcommand{\tr}[4][]{\ensuremath{{#2}\yrightarrow{#3}[-2pt]_{#1}{#4}}}
\DeclareMathOperator{\trnw}{Post}

\def\No{\mathbb{N}_{0}}
\def\N{\mathbb{N}}
\def\R{\mathbb{R}}

\def\Is{\mathcal{I}}

\def\Rs{\mathcal{R}}

\def\Ss{\mathcal{S}}

\def\P{\mathcal{P}}

\def\xv{\boldsymbol{x}}

\def\uv{\boldsymbol{u}}
\def\wv{\boldsymbol{w}}

\def\O{{\normalfont\textbf{0}}}

\def\edge#1{\yrightarrow{#1}[-2pt]}

\newcommand{\fakeparagraphnospace}[1]{\vspace{1mm}\noindent\textbf{#1.}}
\newcommand{\fakeparagraph}[1]{\fakeparagraphnospace{#1}}

\begin{document}
	
	\title{A Simpler Alternative: Minimizing Transition Systems Modulo Alternating Simulation Equivalence}
	
	\author{Gabriel de A.~Gleizer}
	\email{g.gleizer@tudelft.nl}
	\affiliation{%
		\institution{TU Delft}
		\streetaddress{Mekelweg}
		\city{Delft}
		\country{The Netherlands}}
	
	\author{Khushraj Madnani}
	\email{K.N.Madnani-1@tudelft.nl}
	\affiliation{%
		\institution{TU Delft}
		\streetaddress{Mekelweg}
		\city{Delft}
		\country{The Netherlands}}
	
	\author{Manuel Mazo Jr.}
	\email{m.mazo@tudelft.nl}
	\affiliation{%
		\institution{TU Delft}
		\streetaddress{Mekelweg}
		\city{Delft}
		\country{The Netherlands}}
	
	\begin{abstract}
		This paper studies the reduction (abstraction) of finite-state transition systems for control synthesis problems. We revisit the notion of alternating simulation equivalence (ASE), a more relaxed condition than alternating bisimulations, to relate systems and their abstractions. As with alternating bisimulations, ASE preserves the property that the existence of a controller for the abstraction is necessary and sufficient for a controller to exist for the original system. Moreover, being a less stringent condition, ASE can reduce systems further to produce smaller abstractions. We provide an algorithm that produces minimal AS equivalent abstractions. The theoretical results are then applied to obtain (un)schedulability certificates of periodic event-triggered control systems sharing a communication channel. A numerical example illustrates the results.
	\end{abstract}
	\maketitle

	\section{INTRODUCTION}

    Control synthesis for finite transition systems (FTS), the problem of finding a controller (a strategy) that enforces specifications on a closed-loop system, is a long investigated problem \cite{ramadge1989control}. Supervisory control, as it is often also referred to, has many applications in e.g. automation of manufacturing plants, traffic control, scheduling and planning, and control of dynamical and hybrid systems \cite{cassandras2008introduction, tabuada2009verification}. The clearest advantage of using finite transition systems to model a control problem is that a large class of control problems in finite transition systems are \emph{decidable}, meaning that the controller can be obtained automatically through an algorithm, or that a definitive answer that no controller can enforce the specifications is obtained. The disadvantage is often a very practical one: the problem may be too large to be solved in practice, owing to the large number of states and transitions the control problem may have. In particular, this is the case of scheduling the transmissions of event-triggered control (ETC) systems in a shared network \cite{mazo2018abstracted}, whose traffic models can be abstracted as FTSs \cite{gleizer2020scalable}: often it is not possible to synthesize schedulers for more than a handful of ETC systems, due to the state explosion of the composed system.
    This state-space explosion problem is pervasive, and thus significant attention has been devoted to reducing transition systems. The reduction requires a formal relation between original and reduced system; for verification purposes, the most well-known relation is that of \emph{simulation} \cite{milner1971algebraic, baier2008principles}. Algorithms to reduce systems modulo simulations soon emerged: the first being a reduction modulo \emph{bisimulation}, where algorithms using \emph{quotient systems} are often used \cite{baier2008principles}; later, minimization modulo \emph{simulation equivalence} was devised in \cite{bustan2003simulation}. Simulation equivalence is a weaker relation than bisimulation but allows to verify most of the same properties; in particular, any linear temporal logic (LTL) property that can be verified on a system also holds for a simulation equivalent system.%
    \footnote{Larger classes of logic properties can be verified, such ACTL*, ECTL*, ECTL, ACTL as its sublogics, see \cite{bustan2003simulation}. For control, we are typically interested in LTL specifications.}
    
    For control synthesis, reducing the system using mere simulation notions is not enough. Control synthesis can be seen as a game over a finite alphabet, where the controller plays against an antagonistic environment, and simulations preserve all possible moves from both players, \emph{including moves that are irrational} for the game. The notion that appropriately captures the game aspect of control synthesis problems is that of \emph{alternating simulation,} introduced for multi-agent systems by Alur et al.~in \cite{alur1998alternating}.
    Surprisingly, though, there has been little investigation of the problem of reducing systems modulo \emph{alternating bisimulations} or \emph{alternating simulation equivalence}. Reducing systems using alternating simulation notions has many practical benefits: not only the synthesis problems become smaller, and by extension the obtained controllers, making them easier to implement in limited hardware; but it becomes even more important, we argue, when solving control synthesis problems on a parallel composition of systems, one classic example being scheduling. In this case, the size of the game grows exponentially with the number of systems to be scheduled, hence any reduction on the individual systems results in an exponential reduction of the size of the composed game.
    
    In this work we present a novel algorithm to reduce systems w.r.t.~alternating-simulation equivalence (ASE), a different and relaxed notion than the more popular relative \emph{alternating-bisimulation relation}. ASE is nonetheless stronger than alternating simulation relations, as it guarantees not only that controllers can be transferred from abstraction (the reduced system) to concrete (the original system), but also that non-existence of a controller in the abstraction implies non-existence of a controller for the concrete system. Hence the reduction via ASE is sound and complete for control synthesis. We prove that our algorithm in fact obtains a minimal system that is alternating-simulation equivalent to the original. The algorithm is composed of five steps: (i) computing the maximal alternating simulation relation from the system to itself; (ii) forming the quotient system; (iii) eliminating irrational and/or redundant actions from the controller; (iv) eliminating irrational transitions from the environment; and (v) deleting  states which are inaccessible from any of the initial states. The complexity of the algorithm is $O(m^2)$, where $m$ is the number of transitions in the system to be reduced. This result is a very interesting theoretical contribution on its own right, generalizing the results in \cite{bustan2003simulation}. Because these simulation relations are closed under composition, the presented algorithm has a strong practical relevance for synthesis over composed systems. We demonstrate these benefits  on a case study --- one which in fact motivated the investigation of our problem: scheduling of multiple periodic event-triggered control (PETC) \cite{aastrom2002comparison, tabuada2007event, heemels2013periodic} systems on a shared channel. The insights from our algorithm allow to prove that, under some conditions, ETC and self-triggered control (STC, \cite{velasco2003self, anta2008self, mazo2010iss}) are equally schedulable. Additionally, we use our algorithm on a numerical case study, obtaining in the best case a system 50x smaller than the original one. This resulted in a reduction in CPU time of the scheduling problem of several orders of magnitude in some cases. Furthermore, the reduced systems also provide important insights to the user, as the reduced system indicates somehow the bottlenecks that must be addressed to improve schedulability.
    
	\subsection{Related Work}
	Algorithms for reducing state space preserving bisimulation using quotient systems have been extensively studied \cite{kanellakis1990ccs, lee1992online}, see \cite{baier2008principles, bergstra2001handbook} for an overview. For many practical results, simulation equivalence, a coarser equivalence relation, is preferable. Various algorithms to obtain quotients based on simulation equivalence have been proposed,e.g.,~\cite{henzinger1995computing, ranzato2007lics}, as well as their associated quotients \cite{cleaveland2001equivalence}. However, unlike bisimulation, creating quotients based on simulation equivalence does not result in minimization \cite{bustan2003simulation}.
	Our algorithm and results are akin to those of \cite{bustan2003simulation}; we have here a generalization of its results, as alternating simulation reduces to simulations if one of the players has only one choice in every state.
	
	The reduction of systems using alternating simulation equivalence has been addressed in \cite{majumdar2003symbolic, henzinger2005classification}. Different from the current work, Majumdar et al.~ propose a semi-algorithm that aims at reducing infinite systems into finite systems (not necessarily minimal); instead, here we want to minimize finite systems by reducing the number of states and transitions. These two approaches are complimentary and can be used in combination to obtain minimal finite realizations of certain classes of infinite systems (namely, class 2 systems as per \cite{majumdar2003symbolic}).
	
	Reduction of other types of finite transition systems has been addressed, as in, e.g., \cite{fritz2002state} for alternating B\"uchi automata modulo different notions of simulations, namely direct, fair, and delayed simulations. Although such automata also represent games, they are defined differently than what is usual for control: an alternating B\"uchi automaton accepts a word if the controller can ensure it by playing against the environment; every such word forms the language of the automaton, and simulations must preserve this language in some sense. This is fundamentally different than most control problems, where one is not interested in specific words, but rather that the set of all words generated by the system satisfies some specifications. In addition, \cite{fritz2002state} does not contain results on minimality.
	
	\subsection{Notation}
	
		We denote by $\No$ the set of natural numbers including zero, $\N \coloneqq \No \setminus \{0\}$, $\N_{\leq n} \coloneqq \{1,2,...,n\}$.  
		For a relation $R \subseteq X_a \times X_b$, its inverse is denoted as $R^{-1} = \{(x_b, x_a) \in X_b \times X_a \mid (x_a, x_b) \in R\}$. Every function $F: X_a \mapsto X_b$ can be read as a relation, namely $\{(x_a, x_b) \in X_a \times X_b \mid x_b = F(x_a)\}.$

	\section{Preliminaries}
	\subsection{Labelled Transition Systems}
     A (finite) LTS is a 6-tuple $\Ss \coloneqq (X,X_0,U,Y,\delta,H)$, where $X$ is a (finite) set of states, $X_0 \subseteq X$ is the set of initial states, $U$ is the (finite) set of edge labels called inputs or actions, $Y$ is the set of outputs or observations, $\delta \subset X \times U \times X$ is the set of transitions and $H:X \mapsto Y$, the output map, maps states to their corresponding outputs. %
     Figure \ref{fig:twexample} shows one example of a finite LTS, which is our running example throughout this paper; its meaning is going to be explained in Section \ref{sec:casestudy}. 
     
     \begin{figure}
    \centering
    \footnotesize
    \begin{tikzpicture} [node distance = 1.75cm, on grid, auto]
        \node (q00) [state with output, initial] {$q_{0,1}$ \nodepart{lower} \texttt{T}};
        \node (q01) [state with output, right = of q00] {
            $q_{0,2}$ \nodepart{lower} \texttt{W}    };
        \node (q10) [state with output, initial, below = of q00] {$q_{1,1}$ \nodepart{lower} \texttt{T}};
        \node (q11) [state with output, right = of q10] {
            $q_{1,2}$ \nodepart{lower} \texttt{W}};
        \node (q12) [state with output, right = of q11] {
            $q_{1,3}$ \nodepart{lower} \texttt{W}};
        \node (q13) [state with output, right = of q12] {
            $q_{1,4}$ \nodepart{lower} \texttt{W}};
            
        \path [-stealth, thick]
            (q00) edge node {\texttt{w}}   (q01)
            (q00) edge [loop above]  node {\texttt{s}} (q00)
            (q01) edge [bend right] node[above] {\texttt{w, s}}   (q00)
            (q01) edge node {\texttt{w, s}}   (q10)
            (q10) edge node {\texttt{w}}   (q11)
            (q10) edge [loop above]  node {\texttt{s}} (q10)
            (q11) edge [bend left] node[above] {\texttt{s}}   (q10)
            (q11) edge node {\texttt{w}}   (q12)
            (q12) edge node {\texttt{w}}   (q13)
            (q12) edge [bend left] node {\texttt{w, s}}   (q10)
            (q13) edge [bend left] node {\texttt{w, s}}   (q10);
    \end{tikzpicture}
    \caption{A finite LTS representing a PETC traffic model with scheduler actions. Node labels are states (top) and their outputs (bottom), and edge labels are actions.}
    \label{fig:twexample}
\end{figure}
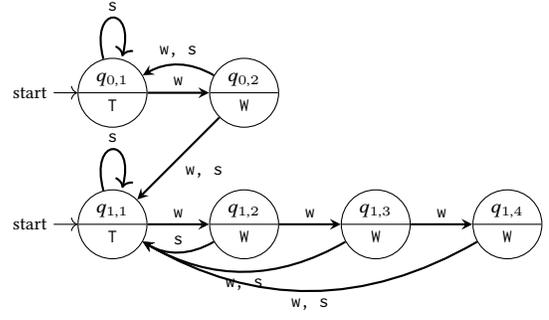

     The \emph{size} of an LTS, denoted by $|\Ss|,$ is the triplet $(|X|, |X_0|, \allowbreak |\delta|)$. This induces a partial order amongst systems sizes using the natural extension of $\leq$ on numbers, i.e., $|(X,X_0,U,Y,\delta,H)| \leq |(X',X'_0,U',Y',\delta',H')|$ iff $|X|<=|X'|$, $|X_0|\leq|X'_0|,$ and $|\delta|\leq|\delta'|$.
     For any $u \in U$ and $x,x' \in X$, We use $\tr{x}{u}{x'}$ to denote the fact that $(x,u,x') \in \delta$. We denote by $U(x) \coloneqq \{u \in U \mid \exists x' \in X, \tr{x}{u}{x'}\}$ the set of input labels available at state $x$, $\Post(x,u) \coloneqq \{x' \in X \mid \tr{x}{u}{x'} \}$ the set of $u$-successors of $x$ and $\Pre(x,u) \coloneqq \{x' \in X \mid \tr{x'}{u}{x} \}$. 
     When the system $\Ss$ is not clear from context, we use, respectively, $\tr[\Ss]{x}{u}{x'}, \Post^\Ss(x,u),$ and $U^\Ss(x)$. 
     System $\Ss$ is said to be deterministic if for every $x \in X$ and $u \in U(x)$, we have $|\Post(x,u)| = 1.$ For a state $x \in X$, we denote by $\Ss(x) \coloneqq \{X, {x}, U, Y, \delta, H\}$ the system $\Ss$ initialized at $x$.
     
     Finite LTSs represent dynamical systems that evolve in discrete state spaces upon the occurrence of actions or events in $U$. They can represent computer programs, machines or factories, but also infinite dynamical systems through the method of abstractions, see \cite{tabuada2009verification}. The problem of control design in finite LTSs is to design a \emph{controller} or \emph{strategy} that chooses the action in $U$ at any point of the run $r$ of the system such that a given specification $\phi$ is satisfied. This has a game aspect in that the controller must ensure $\phi$ no matter what the \emph{environment} does; hence, one can see the environment, i.e., the entity that picks transitions in $\delta$ given the outbound state $x$ and the control action $u$, as antagonist to the controller objectives. The specification $\phi$ is typically given in terms of linear temporal logic (LTL), from which two popular particular cases are safety and reachability. In our scheduling case study (§\ref{sec:casestudy}), we have a safety problem, which is to avoid collisions during transmissions over a shared communication channel.

\subsection{Alternating Simulation and Equivalence}

The concept of alternating simulations was first proposed by \cite{alur1998alternating} for multi-player games on structures called alternating transition systems. It was later simplified by Tabuada for a two-player game, where the \emph{controller} chooses actions in $U$ to meet some specification against an antagonist \emph{environment} that chooses the transitions. The following definition is an adaptation of Tabuada's \cite{tabuada2009verification}:
\begin{definition}[Alternating simulation (AS)]\label{def:altsim}
		Consider two systems $\Ss_a \coloneqq (X_a,X_{a0},U_a,Y_a,\delta_a,H_a)$ and $\Ss_b = (X_b,X_{b0},U_b,Y_b,\delta_b,\allowbreak,H_b)$. We say that $\Ss_b$ is an \emph{alternating simulation} of $\Ss_a$, denoted by $\Ss_a \preceq_{AS} \Ss_b$, if there exists a relation $R \subseteq X_a \times X_b$ satisfying following requirements:
		\begin{enumerate}[label=(\roman*)]
			\item $\forall x_{b0} \in X_{b0} ~\exists x_{a0} \in X_{a0}$ such that $(x_{a0},x_{b0}) \in R$;
			\item $\forall (x_a,x_b) \in R$, it holds that $H(x_a) = H(x_b)$;
			\item $\forall (x_a,x_b) \in R, \forall u_a \in U_a(x_a) ~\exists u_b \in U_b(x_b)$ such that $\forall x_b' \in \allowbreak\Post^{\Ss_b}(x_b,u_b),$ $\exists x_a' \in \Post^{\Ss_a}(x_a,u_a)$ s.t.~$(x_a',x_b') \in R.$
		\end{enumerate}
		We call $R$ an \emph{alternating simulation relation (ASR)} from $\Ss_a$ to $\Ss_b$. When using a specific relation $R,$ we use the notation $\Ss_a \preceq_R \Ss_b.$
	\end{definition}
It is easy to see that if two relations $R_1$ and $R_2$ satisfy $\Ss_a \preceq_{R_1} \Ss_b$ and $\Ss_a \preceq_{R_2} \Ss_b$, then $\Ss_a \preceq_{R_1 \cup R_2} \Ss_b$. The union of all ASRs from $\Ss_a$ to $\Ss_b$ is called the \emph{maximal alternating simulation relation} from $\Ss_a$ to $\Ss_b$.

Intuitively, given LTS $\Ss_a$ and $\Ss_b$, an ASR from an LTS $\Ss_a$ to $\Ss_b$, implies that every controller move of $\Ss_a$ can be ``replicated'' by the controller of $\Ss_b$ and every environment move of $\Ss_b$ can be ``replicated'' by that of $\Ss_a$. Informally, this means that the controller of $\Ss_b$ is at least as powerful as that of $\Ss_a$ and the environment of $\Ss_a$ is at least as powerful as that of $\Ss_b$. This interpretation is also behind our modification of the definition w.r.t.~\cite{tabuada2009verification}, where condition (i) is reversed: in our definition, the ``environment'' picks the initial state, so every initial state in $\Ss_b$ must be matched in $\Ss_a$.%
\footnote{Note that Tabuada's definition and ours are not fundamentally different. In both cases, one could have a single initial state, and condition (i) of Def.~\ref{def:altsim} would be a consequence of condition (iii) by adding silent transitions from the initial state to the ``real'' initial state set; for Tabuada's definition, condition (i) would be derived by (iii) if instead the controller would have a different action for each of these transitions.}

The importance of alternating simulations for control stem from the following fact: given any temporal-logic specification $\phi$ over the alphabet $Y$, if $\Ss_a \preceq_{AS} \Ss_b,$ then the existence of a controller for $\Ss_a$ such that the closed-loop system satisfies $\phi$ implies that there exists a controller for $\Ss_b$ meeting the same specification; in fact, the strategy for $\Ss_a$ can be refined for $\Ss_b$. Moreover, for any specification $\phi$, if $(x,x') \in R$ and the controller can ensure $\phi$ from $x$, then it can ensure $\phi$ from $x'$; the symmetric notion holds: if the controller cannot ensure $\phi$ from $x'$, then it cannot ensure it  from $x$. An additional reason for the importance of AS is that it commutes with composition, making this notion suitable for control design of a composition of systems, such as the scheduling problem we tackle in §\ref{sec:casestudy}. For a thorough exposition about these facts and how to synthesize controllers for several types of specifications we refer the reader to \cite{tabuada2009verification}. Here we are interested in reducing a system $\Ss_a$ preserving an \emph{if-and-only-if} property; namely, \emph{for any specification} there exists a controller for $\Ss_a$ iff there exists a controller for $\Ss_b$. The most known notion for this is that of alternating bisimulation:
\begin{definition}[Alternating bisimulation] Two LTSs $\Ss_a$ and $\Ss_b$ are said to be \emph{alternatingly bisimilar}, denoted by $\Ss_a \bisim_{AS} \Ss_b,$ if there is an ASR $R$ from $\Ss_a$ to $\Ss_b$ such that its inverse $R^{-1}$ is an ASR from $\Ss_b$ to $\Ss_a$.
\end{definition}
A relaxed notion w.r.t.~bisimulation is that of equivalence:
\begin{definition}[Alternating simulation equivalence (ASE)] Two LTSs $\Ss_a$ and $\Ss_b$ are said to be \emph{alternating-simulation equivalent}, denoted by $\Ss_a \equivAS \Ss_b,$ if there is an ASR $R$ from $\Ss_a$ to $\Ss_b$ and an ASR $R'$ from $\Ss_b$ to $\Ss_a$.
\end{definition}
ASE reduces to bisimulation when $R' = R^{-1};$ nevertheless, it preserves by definition the if-and-only-if property we are interested in. Moreover, a second relation is an extra degree of freedom to find a reduced system that is ASE to the original. There is a price to pay for this freedom: the controller designed for the reduced system will not be as \emph{permissive} as the best controller that could be created by the original system; in other words, it may contain fewer actions available to pick from at any point in the system's run. Nonetheless, this can be regarded as a benefit, considering the sheer size the strategies for large LTSs can have.


\section{Main result}\label{sec:theory}
In this section, we present our main result:%
\footnote{When a proof is not right after the result statement, see it in the Appendix.} %
given an LTS $\mathcal{S},$ there exists a polynomial time algorithm that constructs a minimal LTS $\mathcal{S}_{\min}$ equivalent to $\mathcal{S}$ modulo alternating simulation (AS). %
That is, $\Ss_{\min} \equivAS \Ss$ and $|\Ss_{\min}| \leq |\Ss'|$ for any $\Ss'$ satisfying $\Ss' \equivAS \Ss$. %
(i) We first provide an overview of the algorithm to obtain such a minimal system.  
(ii) We then provide the details of the each step of the algorithm and prove its correctness by showing that all steps preserve alternating simulation equivalence. (iii) We show that the output of the algorithm is indeed the unique minimum LTS (up to isomorphism) alternating-simulation equivalent to the input LTS $\Ss$. 
This, in turn, implies that for every LTS there is a unique minimum LTS equivalent modulo AS system that can be constructed using our algorithm. 

\subsection{Overview of the algorithm}
The algorithm can be summarized as follows.  For a system $S \coloneqq (X,X_0,U,Y,\delta, H)$, we denote by $\trsize(S):=|\delta|+|X_0|$, a measure for number of transitions in the system \footnote{We add the cardinality of $X_0$ to total number of transitions because in principle the results we use from \cite{chatterjee2012faster} assumes that there is a unique initial state. Multiple initial states can be simulated by adding silent transitions from a dummy initial state to all the states in $X_0$ which requires $|X_0|$ extra transitions}.
Let $|X|= n$ and $\trsize(S) = m$.

\begin{figure}
    \centering
    \small
    \begin{tikzpicture} [node distance = 1cm, on grid, auto]
        \node (q00) {$q_{0,1}$};
        \node (q01) [right = of q00] {$q_{1,1}$};
        \node (q10) [right = of q01] {$q_{0,2}$};
        \node (q12) [above right = of q10] {
            $q_{1,3}$};
        \node (q11) [below right = of q12] {
            $q_{1,2}$};
        \node (q13) [below right = of q10] {
            $q_{1,4}$};
            
        \path [-stealth, thick]
            (q00) edge (q01)
            (q10) edge (q12) edge (q13)
            (q12) edge  (q13)
            (q13) edge   (q12)
            (q12) edge (q11)
            (q13) edge (q11);
    \end{tikzpicture}
    \caption{Maximal alternating simulation relation $\MAS$ for the system in Fig.~\ref{fig:twexample}: $q \edge{} q'$ means that $(q,q') \in \MAS.$ Self-loops and relations implied from transitivity are omitted.} 
    \label{fig:MAS}
\end{figure}
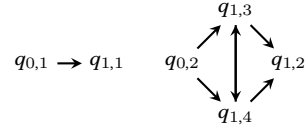
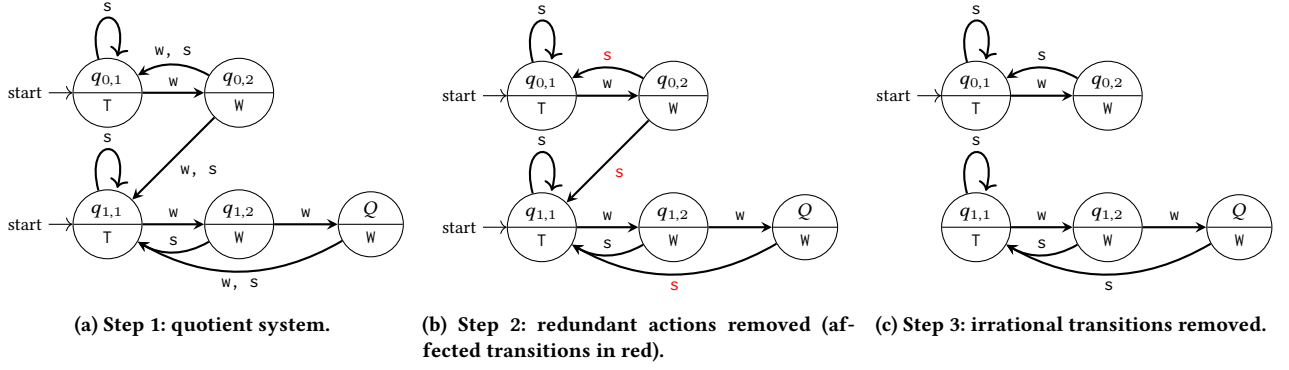
\begin{figure*}[t!]
	\begin{subfigure}[t]{0.32\linewidth}
		\centering
        \footnotesize
        \begin{tikzpicture} [node distance = 1.75cm, on grid, auto]
            \node (q00) [state with output, initial] {$q_{0,1}$ \nodepart{lower} \texttt{T}};
            \node (q01) [state with output, right = of q00] {
                $q_{0,2}$ \nodepart{lower} \texttt{W}    };
            \node (q10) [state with output, initial, below = of q00] {$q_{1,1}$ \nodepart{lower} \texttt{T}};
            \node (q11) [state with output, right = of q10] {
                $q_{1,2}$ \nodepart{lower} \texttt{W}};
            \node (q12) [state with output, right = of q11] {
                $Q$ \nodepart{lower} \texttt{W}};
                
            \path [-stealth, thick]
                (q00) edge node {\texttt{w}}   (q01)
                (q00) edge [loop above]  node {\texttt{s}} (q00)
                (q01) edge [bend right] node[above] {\texttt{w, s}}   (q00)
                (q01) edge node {\texttt{w, s}}   (q10)
                (q10) edge node {\texttt{w}}   (q11)
                (q10) edge [loop above]  node {\texttt{s}} (q10)
                (q11) edge [bend left] node[above] {\texttt{s}}   (q10)
                (q11) edge node {\texttt{w}}   (q12)
                (q12) edge [bend left] node {\texttt{w, s}}   (q10);
        \end{tikzpicture}
        \caption{Step 1: quotient system.}
        \label{fig:S1}
	\end{subfigure}~~
	\begin{subfigure}[t]{0.32\linewidth}
		\centering
        \footnotesize
        \begin{tikzpicture} [node distance = 1.75cm, on grid, auto]
            \node (q00) [state with output, initial] {$q_{0,1}$ \nodepart{lower} \texttt{T}};
            \node (q01) [state with output, right = of q00] {
                $q_{0,2}$ \nodepart{lower} \texttt{W}    };
            \node (q10) [state with output, initial, below = of q00] {$q_{1,1}$ \nodepart{lower} \texttt{T}};
            \node (q11) [state with output, right = of q10] {
                $q_{1,2}$ \nodepart{lower} \texttt{W}};
            \node (q12) [state with output, right = of q11] {
                $Q$ \nodepart{lower} \texttt{W}};
                
            \path [-stealth, thick]
                (q00) edge node {\texttt{w}}   (q01)
                (q00) edge [loop above]  node {\texttt{s}} (q00)
                (q01) edge [bend right] node[above] {\color{red} \texttt{s}}   (q00)
                (q01) edge node {\color{red} \texttt{s}}   (q10)
                (q10) edge node {\texttt{w}}   (q11)
                (q10) edge [loop above]  node {\texttt{s}} (q10)
                (q11) edge [bend left] node[above] {\texttt{s}}   (q10)
                (q11) edge node {\texttt{w}}   (q12)
                (q12) edge [bend left] node {\color{red} \texttt{s}}   (q10);
        \end{tikzpicture}
        \caption{Step 2: redundant actions removed (affected transitions in red).}
        \label{fig:S2}
	\end{subfigure}~~
	\begin{subfigure}[t]{0.32\linewidth}
	    \centering
        \footnotesize
        \begin{tikzpicture} [node distance = 1.75cm, on grid, auto]
            \node (q00) [state with output, initial] {$q_{0,1}$ \nodepart{lower} \texttt{T}};
            \node (q01) [state with output, right = of q00] {
                $q_{0,2}$ \nodepart{lower} \texttt{W}    };
            \node (q10) [state with output, below = of q00] {$q_{1,1}$ \nodepart{lower} \texttt{T}};
            \node (q11) [state with output, right = of q10] {
                $q_{1,2}$ \nodepart{lower} \texttt{W}};
            \node (q12) [state with output, right = of q11] {
                $Q$ \nodepart{lower} \texttt{W}};
                
            \path [-stealth, thick]
                (q00) edge node {\texttt{w}}   (q01)
                (q00) edge [loop above]  node {\texttt{s}} (q00)
                (q01) edge [bend right] node[above] {\texttt{s}}   (q00)
                (q10) edge node {\texttt{w}}   (q11)
                (q10) edge [loop above]  node {\texttt{s}} (q10)
                (q11) edge [bend left] node[above] {\texttt{s}}   (q10)
                (q11) edge node {\texttt{w}}   (q12)
                (q12) edge [bend left] node {\texttt{s}}   (q10);
        \end{tikzpicture}
        \caption{Step 3: irrational transitions removed.}
        \label{fig:S3}
	\end{subfigure}
	\caption{\label{fig:1-3} System of Fig.~\ref{fig:twexample} after steps 1 (left), 2 (middle) and 3 (right), where $Q = \{(q_{1,3}), (q_{1,4})\}$.}
\end{figure*}
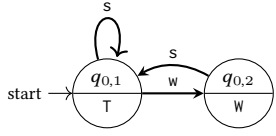
\begin{figure}
    \centering
    \footnotesize
    \begin{tikzpicture} [node distance = 1.75cm, on grid, auto]
        \node (q00) [state with output, initial] {$q_{0,1}$ \nodepart{lower} \texttt{T}};
        \node (q01) [state with output, right = of q00] {
            $q_{0,2}$ \nodepart{lower} \texttt{W}    };
        \path [-stealth, thick]
            (q00) edge node {\texttt{w}}   (q01)
            (q00) edge [loop above]  node {\texttt{s}} (q00)
            (q01) edge [bend right] node[above] {\texttt{s}}   (q00);
    \end{tikzpicture}
    \caption{System of Fig.~\ref{fig:twexample} after step 4. This is a minimal system modulo ASE.}
    \label{fig:S4}
\end{figure}

    \textbf{Step 0: Construct the maximal alternating simulation relation}, denoted by $\MAS$, from $\mathcal{S}$ to itself. This could be constructed using fixed-point algorithms as in \cite{alur1998alternating} or the more efficient algorithm presented in \cite{chatterjee2012faster}, whose complexity is $O(m^2)$. 
    
    \textbf{Step 1: Create a quotient system using $\MAS$} of $\mathcal{S}$ by combining all the equivalent states (and hence all their incoming transitions and outgoing transitions) to get a quotient of the system modulo AS.
    This requires $O(n+m)$ computations given the partition $\mathcal{P}$ (which can be constructed while building $\MAS$) as constructing the quotient transition relation from $\delta$ requires taking the union of all the outgoing transitions from any state in the given partition.
    
   Recall that, if $(q,q') \in \MAS$ then if the controller can meet a specification from the state $q$ then it will definitely meet it from state $q'$. Moreover, if the controller fails to meet the specification from $q'$ it will definitely fail from $q$. In other words, $q'$ (resp.~$q$) is more advantageous position for the controller (resp.~environment) as compared to $q$ (resp.~$q'$). This intuition is central to the next two steps.
    
    \textbf{Step 2: Remove irrational choices and redundant choices for the controller}: For every $x \in X$ and every $a,b \in U(x), a \ne b$ if for every $x_b \in \Post(x,b)$ there exists an $x_a \in \Post(x,a)$ such that $(x_a,x_b) \in$ $\MAS$, then delete all transitions from $x$ on $a$. In other words, remove $a$ from $U(x)$. This is because, for every possible environment move on taking an action $b$ leads to a more (or equally) advantageous state for the controller as compared to any possible state the system can end up on action $a$ by controller.
    To check this, every transition is compared with every other transition at most once. Hence, the complexity of this step in the worst case is bounded by $O(m^2)$.

    \textbf{Step 3: Remove sub-optimal irrational choices for the environment}: For every pair $x_1, x_2 \in X_0$, if $(x_1 ,x_2) \in \MAS$, then the choice of environment to start from $x_2$ will be irrational as $x_1$ is more advantageous position for the environment to start with. Hence, we remove $x_2$ from the initial state set (which is clearly an irrational move for the environment). Similarly, if $(x_1, x_2) \in \MAS$, then for every $a \in U$ if $x' \in \Pre(x_1,a) \cap \Pre(x_2,a)$, remove transition $(x',a,x_2)$ from $\delta$. This is because, if the system is at $x'$, and if the controller chooses an action $a$, the choice of moving to $x_2$ instead of $x_1$ is irrational for the environment as $x_1$ is more advantageous state for the environment. hence, we delete the transition $(x',a,x_2)$.
    Similarly to step 3, before its deletion (or not), any transition is compared with all other transitions at most once. Hence, the worst case complexity is bounded by $O(m^2)$.
    
     \textbf{Step 4: Remove Inaccessible States}: Finally remove all the states that are not accessible from any initial state. This is a routine step with complexity $O(n+m)$. Note that while it seems that Steps 3 and 4 only remove transitions, this does not mean that they do not contribute in the reduction of number of states. Due to the removal of transitions, it could happen that a large fraction of the graph becomes unreachable. This is the step that cashes in the benefit of steps 3 and 4 in terms of reduction in state size. 

The maximal alternating simulation relation from our working example (Fig.~\ref{fig:twexample}) is depicted in Fig.~\ref{fig:MAS}. Figures \ref{fig:1-3} and \ref{fig:S4} illustrate the successive application of each step 1--4 on it.

\subsection{Preserving equivalence Modulo AS: correctness results}

In this section, we formally present the construction/reduction mentioned in each step 1--4 and show that those reductions preserve equivalence modulo AS. These proofs are available in the appendix, due to space limitations. We also present results on the dimension reduction resulting from each step.
We fix $\mathcal{S} \coloneqq (X,X_0, U,Y, \delta, H)$ for this section as a given LTS and apply our reduction steps. For any $i \in \{1,2,3,4\}$ the system resulting of applying step $i$: $S^i(\mathcal{S})$ is the system $S_i$.

\textbf{Step 1: Creating a quotient system}. First, a quotient system $\mathcal{S}_1$ of  $\mathcal{S}$ is created using $\MAS$ as follows. %
Consider the partition $\mathcal{P} = \{Q_1 , Q_2, \ldots, Q_m\}$ of $X$ where each $Q_i$ is the maximal subset of $X$ such that for any states $(p,q) \in Q_i$ it holds that $(p,q) \in  \MAS$ and $(q,p)  \in \MAS$. 

\begin{definition}[Alternating simulation quotient]\label{def:quotient}
   The system $\mathcal{S}_1 \!\coloneqq (X_1,X_{0,1}, U_1, \delta_1,Y, H_1)$ is called the \emph{alternating simulation quotient} of $\mathcal{S}$ w.r.t.~$\MAS$ iff $X_1 = \mathcal{P},$ $X_{0,1} = \{Q \mid Q\in X \wedge \exists q \in Q.~ q \in X_0\}$, $U_1 = U$, $\delta_1 = \{(Q, u, Q') \mid \exists q \in Q.~ \exists q' \in Q.~ (q,u,q)' \in \delta\}$, $\forall Q \in X_1. H_1(Q) = H(q)$ for any $q \in Q$ ($H_1$ is well-defined as $\forall q, q' \in Q.~ H(q) = H(q')$).
\end{definition}

This construction is similar to the celebrated quotient systems used for simulation and bisimulation; here we just make use of the already existing \MAS instead of performing a refinement algorithm, like it has been done for simulation equivalence \cite{bustan2003simulation}. Step 1 preserves equivalence modulo AS: 
\begin{lemma}
\label{lem:step1correctness}
$\mathcal{S} \equivAS \mathcal{S}_1$.
\end{lemma}

Let $\parts:X \mapsto X_1$ be the function that maps every state to its corresponding partition, and $\MAS_1 \subseteq X_1 \times X_1$ be the smallest relation satisfying (I) $\forall (p,q) \in \MAS. (\parts(p), \parts(q)) \in \MAS_1$ and, (II) $\forall (P,Q) \in \MAS_1. \exists p \in P. \exists q \in Q. (p,q)\in \MAS$. Note that $\forall P,Q \in X_1. \exists p \in P. \exists q \in Q.(p,q)\in \MAS \Rightarrow \forall p' \in P. \forall q' \in Q. (p',q') \in \MAS$. This is because every  $P,Q \in X_1$ are sets containing states of $\mathcal{S}$ which are equivalent modulo $\MAS$. Hence, if any element of $P$ is related to any element $Q$ with respect to $\MAS$, then by transitivity of $\MAS$ all elements of $P$ are related to all elements of $Q$. Hence, (II) implies (III) $\forall (P,Q) \in \MAS_1. \forall p \in P. \forall q \in Q. (p,q)\in \MAS$.
The following fact holds:

\begin{lemma}\label{lem:step1mas}
(1) $\MAS_1$ is the maximal ASR from $\Ss_1$ to itself. Moreover, (2) $\MAS_1$ is a partial order.
\end{lemma}

In fact, if $\MAS$ is a partial order (i.e., $(p,q) \in \MAS \implies (q,p) \notin \MAS$ for every $p\neq q$), then step 1 does not affect $\Ss$.

\begin{proposition}\label{prop:c1}
$|X_1| \leq |X|$ and $\trsize(S_1) \leq \trsize(S)$, and no pair $P,Q$ of $X_1$ is equivalent modulo AS. Moreover, if $\MAS$ is not antisymmetric, then $|X_1| < |X|$.
\end{proposition}

\textbf{Step 2: Removing irrational and redundant controller\\ choices}. 
We construct $\mathcal{S}_2 \coloneqq (X_2,X_{0,2}, U_2, \delta_2,Y_2, H_2)$ from $\mathcal{S}$ as follows. $X_2 = X$, $X_{0,2} = X_0$, $U_2 = U$, $Y_2 = Y$, $H_2 = H$. Before defining $\delta_2$, we define an ordering $\sqsubseteq_{\mathcal{S}}$ on elements of $X \times U$:
 $(p',u') \sqsubseteq_{\mathcal{S}} (p,u) \iff u \in U(p) \wedge \forall (p,u,q) \in \delta. \exists (p',u',q') \in \delta. (q',q) \in \MAS$. Note that $\sqsubseteq$ is a transitive relation. We say that an action $u'$ is an \emph{irrational move at a state $p$} of an LTS $\mathcal{S}$ iff $\exists u.  (p,u) \sqsubseteq_{\mathcal{S}} (p,u') \wedge \neg ((p,u') \sqsubseteq_{\mathcal{S}} (p,u))$. State $p$ in this case is said to have \emph{irrational moves}.
 Similarly, we say that $u, u'$ are \emph{equally rational} at a state $p$ of an LTS iff $(p,u) \sqsubseteq_{\mathcal{S}} (p,u') \wedge ((p,u') \sqsubseteq_{\mathcal{S}} (p,u))$. Moreover, if $u$ and $u'$ are distinct then the state $p$, in this case, is said to have \emph{redundant moves}.
 We construct $\delta_2$ by removing all the transitions on irrational actions at $p$. Followed by this, we make available only one of the equally rational actions. 
This procedure preserves equivalence modulo AS. Let $\mathcal{I}:X\mapsto X$ be the identity function.
\begin{lemma}\label{lem:step2correctness}
 $\mathcal{S}_2\preceq_{\mathcal{I}}\mathcal{S}\preceq_{\MAS}\mathcal{S}_2$. Hence,   $\mathcal{S} \equivAS \mathcal{S}_2$.
\end{lemma}

\begin{proposition}
\label{prop:c2}
$|X_2| = |X|,$ $\trsize(\mathcal{S}_2) \leq \trsize(\mathcal{S}),$  and for every state $q \in X_2,$ $U_2(q)$ only contains non-redundant rational actions. Moreover, if there are irrational or redundant actions available from any state $q$ in $\mathcal{S}$, then $\trsize(\mathcal{S}_2) < \trsize(\mathcal{S})$.
\end{proposition}

\textbf{Step 3: Eliminating Irrational Choices for Environment}. 
 We construct $\mathcal{S}_3 \coloneqq (X_3,X_{0,3}, U_3, \delta_3,Y_3, H_3)$ from $\mathcal{S}$ as follows. $X_3 = X$, $U_3 = U$, $Y_3 = Y$, $H_3 = H$. Before the construction of $X_{0,3}$ and $\delta_3$ we define a new relation amongst transitions: any transition $(p,u,q')$ in $\delta$ is called a \emph{younger sibling} of a transition $(p,u,q)$ in $\delta$ with respect to $\mathcal{S}$ iff $(q, q') \in \MAS \wedge (q',q) \notin \MAS$. Similarly, an initial state $q_0'$ is called a younger sibling of yet another initial state $q_0$ with respect to $\mathcal{S}$ iff $(q_0, q_0') \in \MAS \wedge (q_0',q_0) \notin \MAS$. 
 Then, $X_{0,3}$ and $\delta_3$ are constructed from $X_0$ and $\delta$ by deleting all the younger siblings. In other words, given any state $p$ and $u \in U(q)$, if there are two transitions $(p,u,q')$ and $(p,u,q)$ in $\delta$ and if $q \leAS q'$ but not vice-versa (i.e., $q$ is strictly more advantageous position for the environment as compared to $q'$) then delete the transition $(p,u,q')$ from $\mathcal{S}$, as the environment has no reason to choose $q'$ over $q$. Note that this definition is similar to the \emph{younger brother} definition of \cite{bustan2003simulation}, but here we need to take the label of the transitions into account while defining the ``sibling'' relationship due to the definition of AS.
\begin{lemma}\label{lem:step3correctness}
$\mathcal{S}\preceq_{\mathcal{I}}\mathcal{S}_3\preceq_{\MAS}\mathcal{S}$. Thus, $\mathcal{S}_3\equivAS\mathcal{S}$.
\end{lemma}

\begin{proposition}
\label{prop:c3}
$|X_3| = |X|, \trsize(\mathcal{S}_3) \leq \trsize(\mathcal{S}),$ and $\Ss_3$ contains no transitions or initial states that are younger siblings of another transition or initial state, respectively. Moreover, if there is any younger sibling transition or initial state in $\Ss,$ then $\trsize(\mathcal{S}_3) < \trsize(\mathcal{S})$.
\end{proposition}

\textbf{Step 4: Removing states inaccessible from initial state set $X_0$ in $\Ss$}.  Let $X_{\infty}$ be the set of such states inaccessible from any initial state in $X_0$. Then $\mathcal{S}_4 \coloneqq (X_4,X_0, U, \delta_4,Y, H)$, where $X_4 = X \setminus X_{\infty}$, $\delta_4 = \delta \cap (X_4 \times U_4 \times X_4)$.
\begin{lemma}\label{lem:step4correctness}
$\mathcal{S}_4 \bisim_{AS} \mathcal{S}$
\end{lemma}

\begin{proposition}
\label{prop:c4}
$|X_4| \leq |X|,$ $\trsize(\mathcal{S}_4) \leq \trsize(\mathcal{S}),$ and all states in $\Ss_4$ are accessible from $X_{0,4}$. Moreover, if $X_{\infty}$ is non-empty then $|X_4| < |X|$.
\end{proposition}

The combination of Lemmas \ref{lem:step1correctness}, \ref{lem:step2correctness}, \ref{lem:step3correctness} and \ref{lem:step4correctness} gives our main correctness result:

\begin{theorem}\label{thm:correct}
$\Ss \equivAS S^4(S^3(S^2(S^1(\Ss)))).$
\end{theorem}

\subsection{Optimality results}

\begin{theorem}[Necessary Condition for Minimal Equivalent System modulo AS ]
\label{thm:nec}
Given any LTS $\mathcal{S}$, a minimal LTS $\smin \coloneqq(X_{min},X_{0,min}, U_{min} \delta_{min},Y_{min}, H_{min})$ equivalent to the former modulo AS necessarily satisfies the following conditions:
\begin{enumerate}
    \item [$N_1$]For any $p,q \in X_{min}$, $(\smin(p) \equivAS \smin(q)) \Rightarrow p=q$. That is, no two distinct states are equivalent modulo AS to each other.
    \item [$N_2$]For any $p \in X_{min}$, $p$ does not have any irrational or redundant moves.
    \item [$N_3$]$ \nexists\, t_1, t_2 \in \delta_{min}, x_1, x_2 \in X_{0,min}$ such that $t_1$ is a younger sibling of $t_2$ or $x_1$ is a younger sibling of $x_2$.
    \item [$N_4$] All the states in $X_{min}$ are connected from some $x_0 \in X_{0,min}$.
\end{enumerate}
\end{theorem}
\begin{proof}
This theorem is a consequence of Propositions \ref{prop:c1},\ref{prop:c2},\ref{prop:c3}, \ref{prop:c4}. If any condition $i\in \{1,2,3,4\}$ is violated by $\smin$, Step $i$ can be applied to get a  strictly smaller system preserving equivalence modulo AS which contradicts that $\smin$ is minimal.
\end{proof}
\begin{lemma}
\label{lem:necsat}
$\mathcal{S}_{out} = S^4(S^3(S^2(S^1(\mathcal{S}))))$, satisfies the necessary conditions in Theorem \ref{thm:nec}.
\end{lemma}
By Proposition \ref{prop:c1}, we know that after step 1 we get a $S^1(\mathcal{S})$ that satisfies $N_1$. 
The proof then shows that after performing each step $i$, we get a system satisfying $N_i$. Moreover, if the input to the system satisfied any of the previous properties, they will continue to respect it.

We call any LTS satisfying the conditions in Theorem \ref{thm:nec} as \emph{potentially minimal systems}.

In the following we show that the conditions in Theorem \ref{thm:nec} are also sufficient for minimality modulo ASE. 
In fact, we prove something stronger: such a minimal system is unique up to a variant of isomorphism which we introduce as bijective alternating bisimulation isomorphism (BABI). We show this by proving that any two potentially minimum systems $\Ss_1$ and $\Ss_2$ such that $\Ss_1  \equivAS \Ss_2$ implies that they are BABI to each other.
It is important to note that for two structures to be connected via a BABI implies the existence of a bijective alternating bisimulation relation, but the converse is not necessarily true. Hence, the former is stricter than the latter. In fact, the existence of a bijective alternating bisimulation does not necessarily preserve the transition size \footnote{Consider single state systems $B_1, B_2$ one with self loop on $a$ and other with two self loops each on $a$ and $\tilde{a}$.}.

\begin{definition}[Bijective Alternating Bisimulation Isomorphism]
\label{def:babi}
Given any two systems $\Ss_j \coloneqq (X_{j},X_{0,j}, U_{j} \delta_{j},Y_{j}, H_{j})$, $j\in \{1,2\}$, we say that $\Ss_1 \isop \Ss_2$ iff there exists a bijective function $\mathcal{A}:X_1 \mapsto X_2$ such that $\forall p \in X_1 . \mathcal{A}(p) = q$ implies:
\begin{enumerate}
    \item $p \in X_{0,1} \iff q \in X_{0,2}$.
    \item $H_1(p) = H_2(q)$. Vertex labelling is preserved.
    \item There exists a bijection $G_{p,q}: U_1(p) \mapsto U_2(q)$ such that $\forall a \in U_1(p). \Post^{\Ss_2}(q,b) = \{\mathcal{A}(p') \mid p' \in \Post^{\Ss_1}(p,a)\}$ where $b=G_{p,q} (a)$.
\end{enumerate}
\end{definition}
Hence, $\Ss_1 \isop \Ss_2$ implies  $|X_0| = |X_1|$ (implied by the existence of bijection $\mathcal{A}$), $|X_{0,1}| = |X_{0,2}|$ (implied by 1 and bijectivity of $\mathcal{A}$), total number of transitions are equal in both $\Ss_1, \Ss_2$ (implied by 3). Hence,  $|X_1| = |X_2| \wedge \trsize(\Ss_1) = \trsize(\Ss_2)$.
\begin{lemma}
\label{lem:minimum}
Let $\Ss_1$ and $\Ss_2$ be any potentially minimal systems. Then, $\Ss_1 \equivAS \Ss_2$ implies $\Ss_1 \isop \Ss_2$. 
\end{lemma}
\begin{proof}
Given potentially minimal systems $j\in {1,2}$,
$\Ss_j \coloneqq (X_{j}, \allowbreak X_{0,j}, \allowbreak, U_{j}, \delta_{j},Y_{j}, H_{j})$, such that $\Ss_1 \equivAS \Ss_2$ we show that $\Ss_1 \isop \Ss_2$. As $\Ss_1 \equivAS \Ss_2$, denote the maximal ASR from $\Ss_1$ to $\Ss_2$ by $\MAS_1$ and that from $\Ss_2$ to $\Ss_1$ by $\MAS_2$. Let $\mathcal{A} \subseteq X_1 \times X_2$ such that  $\mathcal{A} \coloneqq \{(p,q) \mid (p,q) \in \MAS_1 \wedge (q,p) \in \MAS_2\} = \MAS_1 \cap (\MAS_2)^{-1}$. Note that any pair $(p,q)\in \mathcal{A}$ iff $\Ss_1(p) \equivAS \Ss_2(q)$. We prove the result by showing that $\aA$ is a bijection satisfying all the 3 conditions of the Def. $\ref{def:babi}$. Condition 2 is straightforward: every pair of states occurring in $\aA$ are equivalent modulo Alternating  Simulation and hence have identical labels. 

Now let us focus on Condition 3. We show that $\mathcal{A}$ is a relation satisfying condition 3 of Def. \ref{def:babi}. For that, we construct a relation $G_{p,q}$ satisfying condition 3; then we see it is a bijection. Note that any $(p,q) \in \aA$ implies (C1) $(p,q) \in \MAS_1 \wedge $ (C2) $(q,p) \in \MAS_2$.

\textbf{Construct a candidate relation $G'_{p,q}$ satisfying the consequent of condition 3 of Def.~\ref{def:babi}}.
The former implies (C1.1) for every $a \in U_1(p)$ we can \textit{choose} a $b \in U_2(q)$ such that for every state $q' \in \trnw^{\Ss_2}(q,b)$ we can find a state $p' \in \trnw^{\Ss_1}(p,a)$ such that $p'$ $\leAS$ $q'$ ($(p',q') \in \MAS_1$). 
\\(C2) and (C1.1) together imply (C2.1) for the $b$ chosen in previous step (C1.1) we can \textit{find} an $a' \in U_1(p)$ such that for every state $p'' \in \trnw^{\Ss_1}(p,a')$ we can find a state $q'' \in \trnw^{\Ss_2}(q,b)$ such that $q''$ $\leAS$ $p''$ ($(q'',p'') \in \MAS_2$. 

Combining (C1.1 and C2.1) we get (C3.1)$\forall a \in U_1(p). \exists b \in U_2(q).$ $\exists a' \in U_1(p)$ such that for every state $p'' \in \trnw^{\Ss_1}(p,a')$ there exists a state $q'' \in \trnw^{\Ss_2}(q,b)$ such that $(q'',p'') \in \MAS_2$. Moreover, for this $q''$ we can find a state $p' \in \trnw^{\Ss_1}(p,a)$ such that $(p',q'') \in \MAS_1 (p'\leAS q'')$. Hence, by transitivity of alternating simulation pre-order, for every state $p'' \in \trnw^{\Ss_1}(p,a')$ there exists a state  $p' \in \trnw^{\Ss_1}(p,a)$ such that $p'' \leAS p'$. Hence, $(p,a) \sqsubseteq_{\Ss_1}(p,a')$. Thus, if $a \ne a'$ then $a$ is either redundant or an irrational choice for the controller at state $p$ in LTS $\Ss_1$. This contradicts the assumption that $\Ss_1$ satisfies condition $N_2$. Hence, (C4)$a=a'$.

Thus combining (C3.1) and (C4) we get (C3) for any $a \in U_1(p)$ we can find $b \in U_2(q)$ such that for every state in $p' \in \trnw^{\Ss_1}(p,a)$ we can find $q' \in \trnw^{\Ss_2}(q,b)$ such. that $q'\leAS p'$ ($(q',p') \in \MAS_2$); at the same time, by (C1.1), for every state $q'' \in \trnw^{\Ss_2}(q,b)$ there exists a state in $p'' \in \trnw^{\Ss_1}(p,a)$ such that $p'' \leAS q''$ ($(p'',q'') \in \MAS_1$).

Note that (C3) is equivalent to $\psi(p,q) \coloneqq \forall a \in U_1(p). \exists b \in U_2(q). \varphi(p,q,a,b),$ where
$\varphi(p,q,a,b) = \varphi_1(p,q,a,b) \wedge \varphi_2(p,q,a,b),$
$\varphi_1 \coloneqq
\forall p' \in \trnw^{\Ss_1}(p,a). \exists q' \in \trnw^{\Ss_2}(q,b).  (q',p') \in \MAS_2$
$\varphi_2 \coloneqq 
\forall q'' \in \trnw^{\Ss_{2}}(q,b). \exists p'' \in \trnw^{\Ss_1}(p,a). (p'',q'') \in \MAS_1$.
Let $G'_{p,q} {\subseteq} U_1(p) \times U_2(q)$ such that $(a,b) {\in} G'_{p,q}$ iff $\varphi(p,q,a,b)$ holds. 

{\bf Verify that $G'_{p,q}$ satisfies the consequent of condition 3.} Note that for every $(a,b) \in G'_{p,q}$ we have that every state in $p' \in \trnw^{\Ss_1}(p,a)$ some state in $q' \in \trnw^{\Ss_2}(q,b)$ such that $q' \leAS p'$ ($(q',p') \in \MAS_2$, due to $\varphi_1$),  which in turn, due to $\varphi_2$, satisfies $q' \leAS p''$ for some state $p'' \in \trnw^{\Ss_1}(p,a)$ (i.e, $(p'',q') \in \MAS_1$). 

Now we prove that $p' = p''$ by contradiction. Suppose that $p' \ne p''$. Then, by transitivity of alternating simulation, $p'$ $\leAS$ $p''$. Hence, transition $(p,a,p'')$ is a younger sibling of transition $(p,a,p')$ which contradicts the assumption that $N_3$ is satisfied by $\Ss_1$. Hence (C5) $p' = p''$

Thus, (C5.1) for any $(a,b) \in G'_{p,q}$, for each  $p' \in \trnw^{\Ss_1}(p,a)$ there is a state $q' \in \trnw^{\Ss_2}(q,b)$ such that $p'\leAS q'$ (by $\varphi_1$). Moreover, this $q'$ is in turn alternately simulates $p'$ (by C5 and $\varphi_2$). Hence, $(p',q') \in \aA$. Now we prove that there is a unique $q'$ such that $(p',q') \in \aA.$ (C5.2) Suppose there exists a $p' \in \trnw^{\Ss_1}(p,a)$ that $\succeq$ two distinct states $q',q'' \in \trnw^{\Ss_2}(q,b)$, then by (C5.1) $(p',q') \in \aA$ and $(p',q'') \in \aA$. This would imply that $q'$ and $q''$ are equivalent modulo AS. This contradicts the assumption that $\Ss_2$ satisfies $N_1$.  Hence, for every $p' \in \trnw^{\Ss_1}(p,a)$ there exists a unique $q' \in \trnw^{\Ss_2}(q,b)$ such that $(p',q') \in \aA$. By symmetry of condition $\varphi$, for every $q' \in \trnw^{\Ss_2}(q,b)$ there exists a unique $p' \in \trnw^{\Ss_1}(p,a)$ such that $(p',q') \in \aA$.  

This implies (C6) $\Post^{\Ss_2}(q,b) = \{q' \mid (p',q') \in \aA$ and $p' \in \Post^{\Ss_1}(p,a)\}$. 
Hence, by $\psi(p,q)$ we have (C7) i.e.For any $a \in U_1(p)$ we can find $b \in U_2(q)$ such that $(a,b) \in G'_{p,q}$.

By symmetry, repeating all steps starting from (C2), we get (C8) for any $(p,q) \in \mathcal{A}$ we can construct a relation $G''_{q,p}\subseteq U_2(q) \times U_1(p)$ such that $\Post^{\Ss_1}(p,a) = \{p' | (q',p') \in \aA^{-1} \wedge q' \in \Post^{\Ss_1}(q,a)\}$, reading %
(9) $\forall b \in U_2(q). \exists a \in U_1(p). (b,a) \in G''_{q,p}$.

{\bf Building the bijection $G_{p,q}$.} We now prove that $G_{p,q} \coloneqq G'_{p,q} \cap G''^{-1}_{q,p}$ is a well-defined bijective function such that for any $a\in U_1(p), b\in U_2(q)$, $b = G_{p,q}(a) \implies \Post^{\Ss_2}(q,b) = \{q' \mid (p',q') \in \aA \text{ and } p' \in \Post^{\Ss_1}(p,a)\}$. This proves that $\aA$ satisfies the required condition 3.

(10) For $G_{p,q}$ to be a well-defined function, we need to show that for any $a \in U_1(p)$, there is (A) at least 1 and (B) at most 1 $b \in U_2(q)$ such that $(a,b) \in G_{p,q}$; 
(A) is implied by (C7). 

For (B), assume that for distinct $b_1, b_2 \in U_2(q)$ $(a,b_1), (a,b_2) \in G_{p,q}$. 
By (C6), we get that $\Post^{\Ss_2}(q,b_1) = \{q' \mid (p',q') \in \aA \allowbreak \text{ and } p \in \Post^{\Ss_1}(p,a)\} = \Post^{\Ss_2}(q,b_2)$. But this implies that $b_1$ is a redundant controller choice at state $q$ in LTS $\Ss_2$ which contradicts $N_2$ for system $\Ss_2$. Hence, $G_{p,q}$ is a well-defined function. 
Applying the same reasoning on $G^{-1}_{q,p} = G^{-1}_{p,q} \cap G''_{q,p},$ we get that $G^{-1}_{q,p}$ is also a well-defined function, proving that $G_{p,q}$ is a bijection.

As $G_{p,q}$ contains elements from $G'_{p,q}$, any  $(a,b) \in G_{p,q}$ satisfies (C6).  %
Hence $G_{p,q}$ is the required bijection for condition 3 in Def.~\ref{def:babi}.

\textbf{$\aA$ is a bijection and satisfies condition 1}: (C11) First we show that every initial state is related to a unique initial state. That is, (C11.1 $\aA_0 \coloneqq \aA \cap (X_{0,1} \times X_{0,2})$ is a bijection between $X_{0,1}$ and $X_{0,2}$. We first show by contradiction that $\aA_0$ is a well-defined function. If it is not, then there exists a state $p \in X_{0,1}$ such that (C11.2) either $p$ is not related to any state $q$ in $\aA_0$ or,  (C11.3) $\exists q,q' \in X_{0,2}. (p,q) \in \aA_0 \wedge (p,q') \in \aA_0 \wedge q \ne q'$. Note that $\MAS_2$ is an ASR from $\Ss_2$ to $\Ss_1$, hence from condition (C11.1) of Def.~\ref{def:altsim}, $p$ being an initial state of $\Ss_1$ implies $\exists q'. (q',p) \in \MAS_2$. Now, (due to similar restrictions imposed by condition (C11.1) for $\MAS_1$ being an ASR from $\Ss_1$ to $\Ss_2$) this $q'$ is related with some initial state $p'$ of $\Ss_1$. Hence, $\exists p'. (p',q') \in \MAS_1$. Now note that if $p' = p$, then $(p,q')$ should be in $\mathcal{A}_0$ (by definition) which contradicts the assumption that (C11.2) holds. If $p' \ne p$, we have $\Ss_1(p) \leAS \Ss_2(q')$ $\wedge$ $\Ss_1(q') \leAS \Ss_2,(p) \wedge p \ne p'$. Hence, by transitivity of $\leAS$, $p\leAS p'$, $p\ne p'$ and both are initial states. This implies that $p$ is an initial state which is  younger sibling of $p'$, which contradicts the assumption that $\Ss_1$ satisfies condition $N_3$. Note that to prove $\aA_0$ is a bijection, it suffices to show that $\aA_0^{-1}$ is a well-defined function, which is a symmetrical proof to that of $\aA_0$.

(C12) {\bf Now we show that $\aA$ is a partial function.} That is, every $p\in X_1$ is mapped to a unique $q \in X_2$ via $\aA$. Suppose it is not, i.e., there exists a state $p \in X_{1}$ which is related to two distinct states $q,q' \in X_2$. Hence,  $(p,q),(p,q') \in \aA$. By definition of $\aA$, we have that $(q,p) \in \MAS_2$ and $(p,q') \in \MAS_1$, implying (by transitivity) that $q \leAS q'$; symmetrically, $(p,q) \in \MAS_1$ and $(q',p) \in \MAS_2,$ implying that $q' \leAS q$. Thus, $q$ and $q'$ are equivalent modulo AS which is a contradiction as $\Ss_2$ satisfies $N_1$. Symmetrically, $\aA^{-1}$ is a partial function relation.  

(C13) Note that by (C11) every initial state is mapped to some initial state. By (C12), every state is mapped to a unique state. Hence, every initial state can only be mapped to a \emph{unique} initial state. 

{\bf We now show that $\aA$ (and by symmetry $\aA^{-1}$) is a well-defined function.} 
We already showed that $\aA$ (and $\aA^{-1}$) are partial functions (C12). 
It remains to be proved that a state in $X_1$ can be mapped to at least one state in $X_2$ under $\aA$ (and vice-versa under $\aA^{-1}$). We already showed the latter for states in $X_0$; we now show it for the remaining states. We prove this using contradiction. Assume that there exists a state in $X_1$ that is not mapped to any state in $X_2$ under $\aA$. Let $\P$ be the set of all such states. As $X_1$ is a finite set, so is $\P$. Note that by assumption $N_4$, $\Ss_1$ does not contain any inaccessible state. Hence, every state in $p \in X_1$ can be reached from some initial state in $p_0 \in X_{0,1}$ in $|X_1|$ or less steps. Let $c$ be the minimum number of steps required  to reach the state $p'\in \P$ that is the nearest to the initial state set. That is, no state in $\P$ can be reached in $c-1$ or less steps and there is at least 1 state $p' \in \P$ that is reachable from initial state in $c$ steps. Consider a state $p \in \Pre(p',a)$ for some $a \in U_1$. Because $p$ is reachable in $c-1$ steps, there exists a $q \in X_2$ such that $(p,q) \in \aA$. Now we recover (C5.2): for every $a \in U_1(p). \exists b \in U_2(q). \forall p'' \in \trnw^{\Ss_1}(p,a)$ there exists a unique $q'' \in \trnw^{\Ss_2}(q,b)$ such that $(p'',q'') \in \aA$. This implies that for  $p'$ too there exists a unique $q' \in \trnw^{\Ss_2}(q,b)$ such that $(p',q')\in \aA$. This leads to the contradiction, thus $\aA$ is a well-defined function.
By symmetry, the same holds for $\aA^{-1}$. This implies $\aA$ is a bijection.
\end{proof}

Lemmas \ref{lem:necsat} and \ref{lem:minimum} imply our main optimality result:

\begin{theorem}\label{thm:weminimize}
The system $\mathcal{S}_{out} = S^4(S^3(S^2(S^1(\mathcal{S}))))$ is the unique (up to BABI) minimal system that is ASE to $\Ss$.
\end{theorem}


\section{Case Study: scheduling PETC systems}\label{sec:casestudy}

Event-triggered control (ETC) is an aperiodic sampled-data control paradigm where a \emph{plant} samples its state and sends it to a \emph{controller} upon the occurrence of a designed event. Immediately after, the controller calculates a \emph{control input} that is sent to the actuators of the plant. Despite reducing control-related traffic, ETC's aperiodic traffic makes it challenging to accommodate multiple ETC loops sharing a communication channel: packet collisions are bound to happen, putting the stability of the controlled plants at risk. Therefore, a scheduler must be introduced in the system, in order to adjust the traffic and prevent said collisions, while ensuring stability and performance of the individual plants.

\begin{figure}
    \centering
    \includegraphics[width=\linewidth]{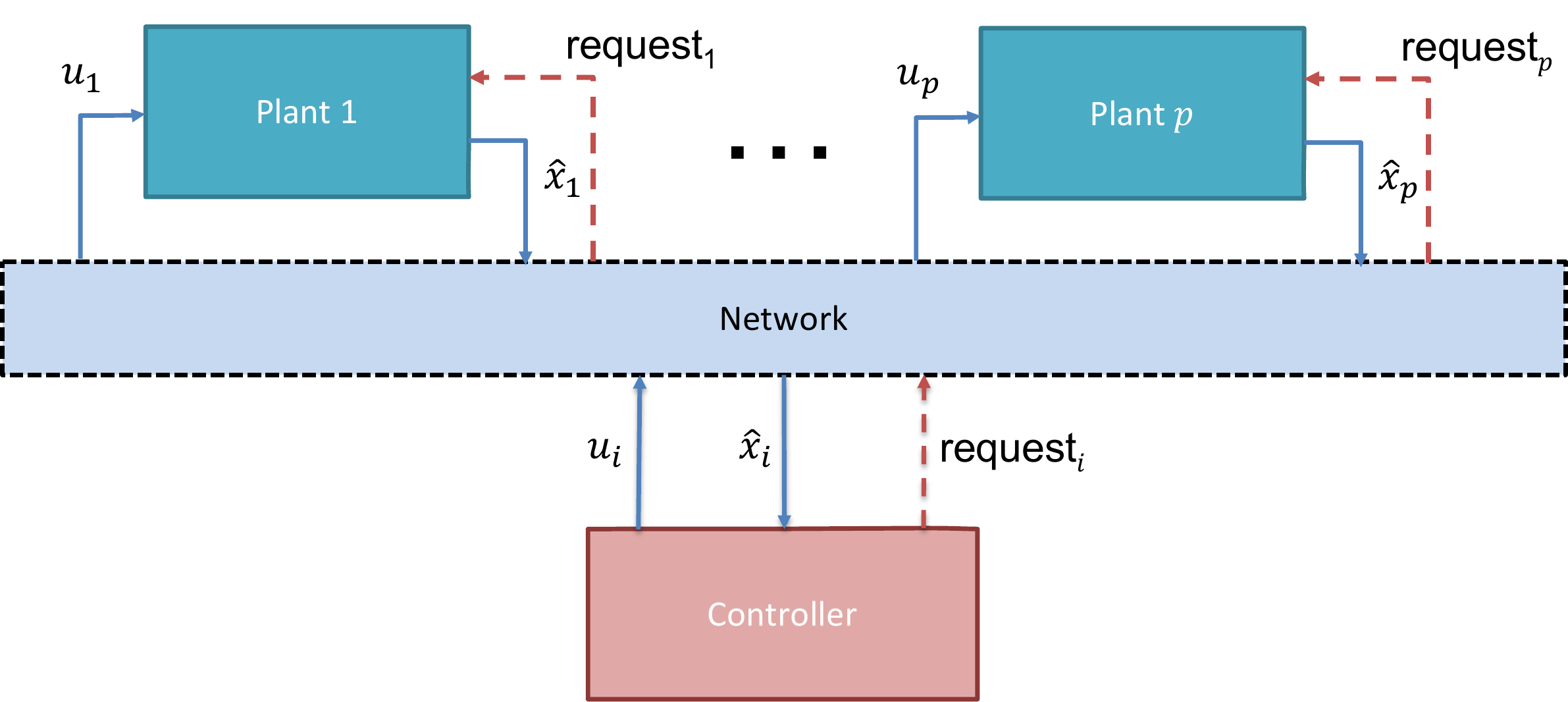}
    \caption{A network of $p$ ETC systems. Plant $i$ can decide (based on the event occurrence) when to send its state sample $\hat\xv_i$ or the controller can request it.}
    \label{fig:scheme}
\end{figure}
Figure \ref{fig:scheme} depicts a networked control system (NCS) with multiple control loops sharing a single communication channel. The plants are described by an ordinary differential equation (ODE), and the controller runs individual \emph{control functions} for each of the plants, as follows:
\begin{equation}\label{eq:plant}
    \begin{aligned}
        \dot\xv_i(t) &= f_i(\xv_i(t), \uv_i(t), \wv_i(t)), \\
        \uv_i(t) &= g_i(\hat\xv_i(t)),
    \end{aligned}
\end{equation}
where $\xv_i(t) \in \R^{n_i}$ is the state of plant $i$, $\uv_i(t) \in \R^{m_i}$ is its control input, and $\wv_i(t) \in \R^{d_i}$ represents the external disturbances that act on it. The variable $\hat{\xv}$ represents the sampled-and-held version of state $\xv$, satisfying
\begin{equation}\label{eq:sampledstate}
    \hat\xv_i(t) =
        \begin{cases}
            \xv(t_{i,k}), & \text{if } t \in [t_{i,k}, t_{i,k+1}), \\
            \O, & \text{otherwise,}
        \end{cases}
\end{equation}
where $t_{i,k} \in \R_+$ represents the $k$-th communication instant for the data of plant $i$. In regular ETC, the communication instants are dictated by a \emph{triggering condition}, such as the seminal one proposed in \cite{tabuada2007event}:
\begin{equation}\label{eq:triggeringtimes}
    \begin{aligned}
    t_{i, k+1} = t^{\mathrm{trigger}}_{i, k+1} &\coloneqq \sup\{t \in h\N \mid t > t_{i,k} \text{ and } \phi(\xv_i(t),\hat\xv_i(t)) \leq 0\}, \\
    \phi(\xv_i(t),\hat\xv_i(t)) &= \norm[\xv_i(t) - \hat\xv_i(t)] - \sigma_i\norm[\xv_i(t)],
    \end{aligned}
\end{equation}
where $\sigma_i \in [0, 1)$ is a design parameter. The parameter $h$ discretizes the time axis, meaning that events can only take place in multiples of $h$. This represents, for simplicity, also the \emph{channel occupancy time}, which is the time it takes for a state measurement $\hat\xv_i(t)$ and the subsequent control action $u_i(t)$ to be sent over the network. In fact, this discretization makes the sampling effectively a periodic event-triggered control, or PETC \cite{heemels2013periodic}.

If multiple control systems operate with communication instants dictated by \eqref{eq:triggeringtimes}, it is generally impossible to prevent communication conflicts in the network; hence we introduce a possibility for the controller to \emph{request} a state sample for any plant \emph{before} its event actually happens. This can prevent collisions, while it is also sound from a control-systems perspective: in ETC, events are designed to happen before an underlying Lyapunov function stops decreasing sufficiently fast, thus ensuring closed-loop stability, see, e.g., \cite{tabuada2007event, heemels2012introduction}. This makes \emph{early sampling} a safe choice from a control performance perspective, and this feature has been extensively exploited in the event-based literature \cite{mazo2010iss, anta2008self}, including in the context of scheduling of ETC systems \cite{gleizer2020scalable}. Therefore, the sampling times $t_{i,k}$ can either occur upon triggering of the condition as in \eqref{eq:triggeringtimes}, or be \emph{requested} earlier by the scheduler, satisfying
\begin{equation}\label{eq:earlysamplingtimes}
t_{i,k+1} \in \{t \in h\N \mid t > t_{i,k} \text{ and } t \leq t^{\mathrm{trigger}}_{i, k+1}\}.
\end{equation}

The quantity $\tau_{i,k} \coloneqq t_{i,k+1} - t_{i,k}$ is called \emph{inter-sample time}. When given these degrees of freedom, the most fundamental question one needs to answer is whether it is possible for a scheduler to coordinate the traffic generated by the $p$ PETC loops while avoiding collisions and ensuring that the communications are timely. We assume that the device that runs the scheduler is capable of listening to all traffic, thus having access to the sampled states of all systems. In fact, this can be the same device that runs the control functions, which is the case depicted in Fig.~\ref{fig:scheme}.

\fakeparagraph{The early-sampling PETC schedulability problem} Consider a network containing $p$ control-loops \eqref{eq:plant} and $C < p$ communication channels with channel occupancy time $h$. Our main goal is to determine whether there exists a strategy that, at every time $t \in h\N$, given the available sampled states $\hat\xv_i(t_{i,k}), \forall i,k$ such that $t_{i,k} < t$, determines which (if any) loops must send their samples to the controller. The number of loops sending their samples must be no greater than $c$, and for each loop $i$, $t \leq t^{\mathrm{trigger}}_{i, k+1}$ must hold; that is, no controller can miss its \emph{deadline} $t^{\mathrm{trigger}}_{i, k+1}$. If a scheduler can be found, we also want to retrieve one such scheduling strategy for real-time implementation.

For simplicity, we assume for the rest of this paper that the time units are selected such that $h = 1$.

\subsection{PETC traffic models as finite-state transition systems}

The problem described above can be seen as a safety control synthesis problem for a hybrid system, which is in general undecidable \cite{alur1995algorithmic, henzinger1995undecidability}. To deal with decidable problems, the control loops $i$ have been abstracted as timed-game automata (TGA) in \cite{kolarijani2015traffic}, and later as regular transition systems in \cite{gleizer2020scalable}, by assuming the same discrete nature of sampling instants as we assume here. For details on how to construct such abstractions, see \cite{gleizer2020scalable, gleizer2021hscc} for linear systems without disturbance, and \cite{delimpaltadakis2021traffic} for general perturbed nonlinear systems. In these abstractions, each state $q$ is a different region $\Rs_q \subset \R^n$ and associated an interval of possible inter-sample times $\{\tau^\mathrm{low}_q, \tau^\mathrm{low}_q + 1, ... \tau^\mathrm{high}_q\}$ at which a trigger can occur. The scheduler can choose to sample earlier than $\tau^\mathrm{low}_q,$ or sample during the aforementioned interval as long as a trigger has not yet occurred. From each state $q$, the set of possible regions reached depends on the chosen inter-sample time $\tau$, regardless of whether the sample is determined by the scheduler or the triggering condition. Hence, the abstraction process outputs a set of transitions $\Delta \subset Q \times T \times Q,$ where $(q,c,q') \in \Delta$ means that $q' \in Q$ can be reached from $q \in Q$ if the inter-sample time is $\tau \in T$. From this, we derive the following definition of PETC traffic model:
\begin{definition}[PETC traffic model]\label{def:trafficmodel} A finite PETC traffic model with scheduler actions is the transition system $\Ss_{\mathrm{PETC}} \coloneqq (X,X_0, \allowbreak \{\texttt{w, s}\}, \delta_\textrm{wait} \cup \delta_\textrm{sched} \cup \delta_\textrm{trigger}, H)$ where
\begin{itemize}
    \item $X = \{(q,c) \mid q \in Q, c \in \{1,2,...,\tau^\mathrm{high}_q\}\},$
    \item $\delta_\textrm{wait} = \{(q,c),\texttt{w},(q,c+1) \mid (q,c) \in X \text{ and } c < \tau^\mathrm{high}_q\},$
    \item $\delta_\textrm{sched} = \{(q,c),\texttt{s},(q',0) \mid (q,c) \in X \text{ and } (q,c,q') \in \Delta\},$
    \item $\delta_\textrm{trigger} = \{(q,c),\texttt{w},(q',0) \mid (q,c) \in X \text{ and } c \geq \tau^\mathrm{low}_q$ \\ $\text{ and } (q,c,q') \in \Delta\},$
    \item $H(q,c) = \texttt{T}$ if $c=0$, or $\texttt{W}$ otherwise.
\end{itemize}
\end{definition}
The actions {\tt w} (for wait) and {\tt s} (for sample) are the scheduler actions; as the spontaneous trigger of a given loop is out of the control of the scheduler, these transitions are considered (adversarial) nondeterminism for the scheduler. This is why the set $\delta_\textrm{trigger}$ is a set of sampling transitions, but they occur when the action {\tt wait} is chosen. The output $T$ represents when a transmission has just occurred, while $W$ means that the loop waited. The initial state depends on the particularities of the scheduling problem and will be discussed later.

Our running example, Fig.~\ref{fig:twexample}, depicts a simple PETC traffic model with only two regions. This example contains only two regions $\Rs_0$ and $\Rs_1$, mapped into $q_0$ and $q_1$, respectively, with $\tau^\mathrm{low}_{q_0} = \tau^\mathrm{high}_{q_0} = 2$, and $\tau^\mathrm{low}_{q_1} = 3$ and $\tau^\mathrm{high}_{q_1} = 4$. The states $(q_1,3)$ and $(q_1,4)$ represent the \emph{triggering phase} of place $q_1$: even if the scheduler decide to wait, the sampling can occur in any of these states.

\subsection{A general result on ETC scheduling}

Using the reduction in Section \ref{sec:theory}, a first general result can be derived for scheduling of PETC.

\begin{definition}[Reduced PETC traffic model]\label{def:reducedtraffic} A reduced PETC traffic model with scheduler actions is the transition system $\Ss'_{\mathrm{PETC}} \coloneqq (X',X_0 \cap X',\{\texttt{w, s}\}, \delta'_\textrm{wait} \cup \delta'_\textrm{sched}, H)$ where
\begin{itemize}
    \item $X' = \{(q,c) \mid q \in Q, c \in \{1,2,...,\tau^\mathrm{low}_q\}\},$
    \item $\delta'_\textrm{wait} = \{(q,c),\texttt{w},(q,c+1) \mid (q,c) \in X' \text{ and } c < \tau^\mathrm{low}_q\},$
    \item $\delta'_\textrm{sched} = \{(q,c),\texttt{s},(q',0) \mid (q,c) \in X' \text{ and } (q,c,q') \in \Delta\},$
    \item $H(q,c) = \texttt{T}$ if $c=0$, or $\texttt{W}$ otherwise.
\end{itemize}
\end{definition}

The difference between Def.~\ref{def:reducedtraffic} and Def.~\ref{def:trafficmodel} is that, in the former, the sampling always happens at most at $\tau^\mathrm{low}_q$ for every $q$, and that this point in time it is a scheduled sampling. In other words, there is no event-based sampling anymore, but the scheduler may decide to sample at the first moment in which it knows that an event trigger could occur. This is very similar in spirit to \emph{self-triggered control} (STC, see \cite{anta2008self, mazo2010iss}), where the controller chooses the sampling time by predicting a worst-case situation in which the event-triggered control would occur. Thus, Def.~\ref{def:reducedtraffic} is can also be regarded as a traffic model for STC systems, again allowing early sampling. Fig.~\ref{fig:twreduced} shows the reduced model from Fig.~\ref{fig:twexample}. The interesting fact is that these two approaches are equivalent from a schedulability perspective:
\begin{proposition}\label{prop:etcisthesameasstc}%
\!\!\footnote{See the proof in the Appendix.} %
The PETC traffic model from Def.~\ref{def:trafficmodel} and its reduced model from Def.~\ref{def:reducedtraffic} are alternating-simulation equivalent, provided $X_0 \subseteq X'.$
\end{proposition}

The interpretation of this result is simple: the choice of waiting at time $\tau^\mathrm{low}_q$ has no advantage over sampling, because in the worst case the environment may choose to sample anyway. Hence, from a schedulability perspective, \emph{ETC brings no benefit over a STC-like sampling strategy that chooses to trigger on the earliest ETC triggering time}. Naturally, this general result does not give the minimal system, which depends on the structure of the particular abstraction, as will be illustrated in the next section.

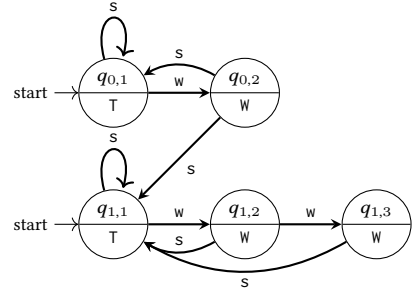
\begin{figure}
    \centering
    \footnotesize
      \centering
      \begin{tikzpicture} [node distance = 1.75cm, on grid, auto]
        \node (q00) [state with output, initial] {$q_{0,1}$ \nodepart{lower} \texttt{T}};
        \node (q01) [state with output, right = of q00] {
            $q_{0,2}$ \nodepart{lower} \texttt{W}    };
        \node (q10) [state with output, initial, below = of q00] {$q_{1,1}$ \nodepart{lower} \texttt{T}};
        \node (q11) [state with output, right = of q10] {
            $q_{1,2}$ \nodepart{lower} \texttt{W}};
        \node (q12) [state with output, right = of q11] {
            $q_{1,3}$ \nodepart{lower} \texttt{W}};
            
        \path [-stealth, thick]
            (q00) edge node {\texttt{w}}   (q01)
            (q00) edge [loop above]  node {\texttt{s}} (q00)
            (q01) edge [bend right] node[above] {\texttt{s}}   (q00)
            (q01) edge node {\texttt{s}}   (q10)
            (q10) edge node {\texttt{w}}   (q11)
            (q10) edge [loop above]  node {\texttt{s}} (q10)
            (q11) edge [bend left] node[above] {\texttt{s}}   (q10)
            (q11) edge node {\texttt{w}}   (q12)
            (q12) edge [bend left] node {\texttt{s}}   (q10);
    \end{tikzpicture}
    \caption{Reduced PETC traffic model of Fig.~\ref{fig:twexample}.}
    \label{fig:twreduced}
\end{figure}

\subsection{Numerical example}\label{ssec:num}

Consider $p$ two-dimensional open-loop-unstable linear systems, borrowed from \cite{tabuada2007event}, of the form \eqref{eq:plant} where
\begin{equation}
    \begin{aligned}
    f_i(\xv_i(t), \uv_i(t), \wv_i(t)) &\coloneqq \begin{bmatrix} 0 & 1 \\ -2 & 3 \end{bmatrix}\xv_i(t) = \begin{bmatrix} 0 \\ 1 \end{bmatrix} \uv_i(t), \\
    g_i(\hat\xv_i(t)) &\coloneqq \begin{bmatrix}1 & -4\end{bmatrix}\hat\xv_i(t).
    \end{aligned}
\end{equation}
The triggering condition is \eqref{eq:triggeringtimes} with $\sigma = 0.7.$ Since all systems have the same model, only one traffic abstraction is needed. We use the abstraction method in \cite{gleizer2021cdc}, where a parameter $l \in \N$ is given to define the depth of the abstraction process: the higher $l$ is, the tighter the simulation relation is w.r.t.~ the original infinite system. Denote by $\Delta_l \subseteq Q_l \times T \times Q_l$ the transition relation from the abstraction using depth $l$, and the resulting PETC traffic model (Def.~\ref{def:trafficmodel}) by $\Ss_l$.

We consider the problem of scheduling on a single channel. From a practical perspective, the scheduling problem requires an initialization phase. When the systems are connected to the network, their states will only be known to the scheduler (and the controller) after the first sample. Because there is only one channel, the timing of the initial transmissions have to be decided by the scheduler, and this timing must be bounded to keep the plant's state under a reasonable distance from its initial value. Let  $T_0$ be this time bound (in number of steps). To model this initialization phase, we append to $\Ss_l$ the states $i_1, i_2, ... i_{T_0}$ and transitions $i_k \edge{\texttt{w}} i_{k+1}$ for all $k < T_0$ and $i_k \edge{\texttt{s}} (q,0)$ for all $k \leq T_0$ and $q \in Q_l$. The initial set is simply $X_0 = \{i_1\}.$ In this example, $T_0$ was set to 10.

\begin{table}\caption{\label{tab:traffic_stats} Size of abstractions before and after minimization, and CPU time to minimize the system (in all cases $|X_0| = 1$).}
	\begin{center}
		\begin{tabular}{cc|cc|cc|cc|c}
		    & & \multicolumn{2}{c}{Original} & \multicolumn{2}{c}{Quotient} & \multicolumn{2}{c}{Minimal} & CPU\\
			 & & $|X|$ & $|\delta|$ & $|X|$ & $|\delta|$ & $|X|$ & $|\delta|$ & time \\
			\hline
			\multirow{3}*{$l$} & 1 & 153 & 832 & 118 & 571 & \bf 11 & \bf 21 & 657 ms \\
			& 2 & 518 & 1879 & 405 & 1566 & \bf 11 & \bf 21 & 8.24 s\\
			& 3 & 683 & 2412 & 604 & 2262 & 587 & 2126 & 15 s\\
			\hline
		\end{tabular}
	\end{center}
\end{table}

We implemented our minimization algorithm in Python and performed the minimization on $\Ss_l, l = 1,2,3.$ The statistics of the traffic model before and after minimization modulo ASE are displayed in Table \ref{tab:traffic_stats}. The additional reduction w.r.t.~only step 1 (quotient system) is evident in all cases. The most interesting phenomenon is the striking reduction of the traffic models for $l=1,2$ to a system with only 11 states and 21 transitions, which is depicted in Fig.~\ref{fig:reallysmall}. Not only this is a massive reduction which greatly simplifies the scheduling problem, it also informs the user that refining the traffic model by increasing $l$ from 1 to 2 is \emph{irrelevant when it comes to schedulability}. As Fig.~\ref{fig:reallysmall} suggests, these traffic models reduce to a single task with recurring deadline of five steps, after the initial phase. Only with $l=3$ more complex behavior can be enforced by the scheduler, which becomes apparent by the fact that the minimization is not so impactfull: 14\% in states and 12\% in transitions. This is to be expected because the original systems we abstract are deterministic, and higher values of $l$ reduce the nondeterminism of the abstraction, giving less room for transition elimination in our algorithm. In all cases, the CPU times are within seconds, with an approximately quadratic dependence on the size of the original system. It is worth noting that our Python implementation uses the naive fixed-point algorithm to get the MAS relation, and this step dominated the CPU time of the reduction. Since the times were satisfactory, no performance optimizations were attempted. 

\begin{figure}
    \centering
    \footnotesize
    \begin{tikzpicture} [->, >=stealth', node distance = 0.7cm, on grid, auto, every state/.style={inner sep=.4mm, minimum size=0, semithick}]
        \node (q00) [state, initial] {\texttt{W}};
        \node (q01) [state, right = of q00] {\texttt{W}};
        \node (q02) [state, right = of q01] {\texttt{W}};
        \node (q03) [state, right = of q02] {\texttt{W}};
        \node (q04) [state, right = of q03] {\texttt{W}};
        \node (q05) [state, right = of q04] {\texttt{W}};
        \node (q06) [state, right = of q05] {\texttt{W}};
        \node (q07) [state, right = of q06] {\texttt{W}};
        \node (q08) [state, right = of q07] {\texttt{W}};
        \node (q09) [state, right = of q08] {\texttt{W}};
        \node (q10) [state, right = of q09] {\texttt{T}};
        
        \path [-stealth, dotted, semithick]
            (q00) edge (q01)
            (q01) edge (q02)
            (q02) edge (q03)
            (q03) edge (q04)
            (q04) edge (q05)
            (q05) edge (q06)
            (q06) edge (q07)
            (q07) edge (q08)
            (q08) edge (q09)
            (q09) edge (q10)
            (q10) edge[bend left] (q06);
            
        \path [-stealth, semithick]
            (q00) edge[bend left] (q10)
            (q01) edge[bend left] (q10)
            (q02) edge[bend left] (q10)
            (q03) edge[bend left] (q10)
            (q04) edge[bend left] (q10)
            (q05) edge[bend left] (q10)
            (q06) edge[bend left] (q10)
            (q07) edge[bend left] (q10)
            (q08) edge[bend left] (q10)
            (q09) edge[bend left] (q10)
            (q10) edge[loop above] (q10);
            
    \end{tikzpicture}
    \caption{Minimized system for the numerical example, $l \in \{1, 2\}$. State labels are their outputs, dashed lines are \texttt{w} actions and full lines are \texttt{s} actions.
    \label{fig:reallysmall}}
\end{figure}
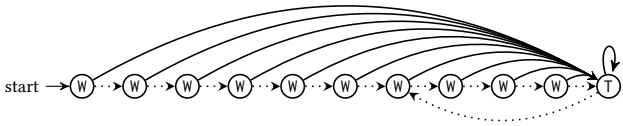

\begin{table}\caption{\label{tab:sched_stats} Scheduler size and CPU time using BDDs.}
	\begin{center}
		\begin{tabular}{cc|cc|cc}
		    & & \multicolumn{2}{c}{Original} & \multicolumn{2}{c}{Minimal} \\
			 $p$ & $l$ & Schedule size & CPU time & Schedule size & CPU time \\
			\hline
			2 & 1 & 3 kB & 3 ms & 894 B & 498 \textmu s \\
			3 & 1 & 7.6 kB & 8 ms & 1.9 kB & 783 \textmu s \\
			4 & 1 & 19 kB & 16 ms & 4.2 kB & 1.4 ms \\
			5 & 1 & 47 kB & 36 ms & 9.5 kB & 2.5 ms \\
			6 & 1 & None & \bf 7.35 s & None & 101 ms \\
			6 & 2 & None & \bf 15.4 min & None & \bf 84 ms \\
			6 & 3 & None & 35.5 min & None & 28.1 min \\
			\hline
		\end{tabular}
	\end{center}
\end{table}

Because of the refinement properties of the abstractions $\Ss_l$ (namely $\Ss_{l+1} \preceq_{AS} \Ss_l \preceq \Ss_{l+1}$), scheduling with these abstractions is sound but not complete: if $p$ ETC plants are detected to be unschedulable for $l,$ one may still find a schedule using a higher value of $l$. Thus, we employed the following scenario: first, set $l=1$ and $p=2$ and increase $p$ until the systems are unschedulable; then, increase $l$ and try again. We used the ETCetera tool \cite{etcetera} to solve the scheduling problem, which has the functionality to create the traffic models $\Ss_l$, perform the parallel composition, and solve the safety game: always avoid a state whose output contains more than one $\mathtt{T}$. Our first attempt used a Python implementation of the composition and safety game solution, where the transitions are encoded with dictionaries. Without minimization and with $p=2$, the scheduling problem took only 801 ms to be concluded, a number close to the 657 ms taken to minimize each system; this is expected, given the quadratic complexity of the minimization algorithm. However, with only $p=3$ \emph{the scheduling problem without reduction crashed due to memory overflow.}%
\footnote{The experiments were performed in a Intel(R) Xeon(R) W-2145 CPU @ 3.70GHz with 31 GB RAM.}%
\emph{After performing the minimization, we were able to compute a scheduler for $p=5$}, a process that took 28.7 min to conclude. With $p=6$, memory overflow also occurred with the minimal systems. Our second attempt to solve the scheduling problem used BDDs to encode the transition systems. Table \ref{tab:sched_stats} summarizes the results of this experiment. As expected, for all cases in $l \in \{1,2\}$ the problem was solved significantly faster with the minimized systems. The difference is much more significant in the non-schedulable cases, which is to be expected because it often requires more iterations in the fixed-point algorithm to detect that no schedule is viable. The difference is particularly massive for $l=2$, owing to the immense reduction of the system dimensions in this case. For the case with $l=3$ the time reduction was not as significant as in the aforementioned cases, which is in par with the smaller system size reduction that was obtained in this case.

\section{Conclusion and Future Work}

We have revisited the notion of alternating simulation equivalence, and argued about the benefits it can bring for size reduction of finite transition systems in the context of controller synthesis. An algorithm was devised to produce minimal abstractions modulo alternating simulation equivalence. The applicability of these theoretical developments was then illustrated in the context of scheduling, providing interesting insights for the analysis of schedulability of event triggered systems.

This work opens the door to several further investigations, in particular: (i) extending the ASE notion to weighted transition systems to produce abstractions preserving quantitative properties; (ii) extensions of these same ideas to timed games; (iii) designing on-the-fly versions of the proposed reduction algorithm; and (iv) implementing symbolically the abstraction algorithm employing binary decision diagrams.

\begin{acks}
	This work is supported by the \grantsponsor{GSERC}{European Research Council}{https://erc.europa.eu/} through the SENTIENT project, Grant No.~\grantnum[https://cordis.europa.eu/project/id/755953]{GSERC}{ERC-2017-STG \#755953}.
\end{acks}

\bibliographystyle{ieeetr} 
\bibliography{mybib} 

\appendix

\newpage
\section{Correctness and reduction proofs}

\begin{proof}[Proof of Lemma \ref{lem:step1correctness}]
We first observe that if any pair of states $(p,q) \in$ $\MAS$ then for all states $p'$ equivalent to $p$ and $q'$ equivalent to $q$ modulo AS, we have that $(p',q') \in \MAS$ (by transitivity of $\MAS$).
Hence $\forall (p,q) \in$ $\MAS$ $\iff$ $\forall p' \in \parts(p), q' \in \parts(q) (p',q') \in \MAS$, where $\parts:X \mapsto X_1$ is the function that maps every state to its corresponding partition. 

To prove the lemma we show that (A) $R_1 \coloneqq \{(p, \parts(q)) \mid (p,q) \in \MAS\}$ is an ASR from $\Ss$ to $\Ss_1$ and (B) $R_2 \coloneqq \{(\parts(p), q) \mid (p,q) \in \MAS\}$ is an ASR from $\Ss_1$ to $\Ss$.

For (A) to be true, $R_1$ must satisfy requirements (i), (ii) and (iii) from Def.~\ref{def:altsim}. Note that $R_1$ contains the set $\{(p, \parts(p) \mid p \in X\}.$ 
By construction, for every state $q \in X_0$, $\parts(q) \in X_{0,1}$ and for every $q \in X$, $H(q) = H_1(\parts(q))$, so requirements (i) and (ii) are satisfied. Assume that $R_1$ does not satisfy requirement (iii), i.e., $\psi \coloneqq \exists q \in X.~ \exists u \in U(q).~ \varphi(q, u)$ where  $\varphi(q,u) \coloneqq \forall u_1 \in U_1(\parts(q)).~ \exists Q' \in \Post^{\Ss_1}(\parts(q),u_1).~ \forall q'\in\Post^{\Ss}(q,u).~ \allowbreak (q',Q') \notin R_1$. 
Now, for any such $Q'$, note that by  construction of $\mathcal{S}_1$, $(\parts(q),$ $u_1,Q') \in \delta_1$ implies that $\exists p \in \parts(q), p' \in Q'.~(p,u_1,p') \in \delta$. However, notice that $(q',Q') \notin R_1 \implies \forall p' \in Q'\,(q',p') \notin \MAS$. Therefore, $\varphi(q, u)$ implies $\varphi'(q,u) \coloneqq \forall u_1 \in \allowbreak U_1(\parts(q)).~ \exists Q' \in \Post^{\Ss_1}(\parts(q),u_1), \exists p' \in Q'.\, \forall q'\in\\\Post^\Ss(q,u).$ $(q',p') \notin \MAS.$

Now let us inspect the set $U_1(\parts(q))$. By definition of $\delta_1,$ $U_1(\parts(q)) \coloneqq \{u \mid u \in U(p), (q,p) \in \MAS, (p,q) \in \MAS\}$ 
Hence, $\forall u_1 \in U_1(\parts(q))\,\exists Q' \in \Post^{\Ss_1}(\parts(q),u_1), \exists p' \in Q'$ can be replaced by $\forall p.(p,q)(q,p)\in\MAS\,\forall u_1 \in U(p).\,\exists p' \in \Post^\Ss(p,u_1)$. Applying this and $\varphi'$ we obtain 
$\varphi'' \coloneqq \forall p$ such that $(p,q),(q,p)\in\MAS.~\forall u_1 \in U(p) .~ \exists p' \in \Post^\Ss(p,u_1) \forall q'\in\Post^\Ss(q,u).~ (q',p') \notin \MAS.$ The formula $\psi$ therefore implies $\exists q \in X.~ \exists u \in U(q).\forall p.\,(p,q),(q,p)\in\MAS.~\forall u_1 \in U(p) .~ \exists p' \in \Post^\Ss(p,u_1).~ \forall q'\in\Post^\Ss(q,u)$ $(q',p') \notin \MAS.$ Finally, we particularize $\forall p$ to  $p = q$ (which is a sound step as $\MAS$ has all reflexive entries,  i.e. $(q,q) \in MAS$), changing $\parts(q)$ to $q$ (this is implied by transitivity of $\leAS$ and $\equivAS$)  and rearrange the initial existential quantifiers to get $\exists q.~ \exists u \in U(q).\,\forall u_1 \in U(q) .~ \exists p' \in \Post^{\Ss}(q,u_1).~ \forall q'\in\Post^{\Ss}(q,u).$ \\$(q',p') \notin \MAS.$ This implies that \MAS (that contains all the reflexive pairs) fails to satisfy requirement (iii) of for system $\Ss$ to itself, which is a contradiction.

For (B), conditions (i) and (ii) also hold trivially. We use contradiction again for condition (iii): suppose $\phi \coloneqq \exists (P,q) \in R_2.\,\exists u_1\in U_1(P)\,\forall u \in U(q)\,\exists q' \in \Post^\Ss(q)\, \underline{\forall P' \in \Post^{\Ss_1}(P,u_1).\,(P',q')}$ $\underline{\notin R_2}.$ By construction of $\delta_1,$ the underlined subformula is equivalent to $\forall p^* \in P.\,u_1\in U(p^*)\, \forall p' \in \Post^{\Ss}(p^*, u_1).\,(\parts(p'),q') \notin R_2.$ By construction of $R_2$, and because all states in $\parts(p')$ are equivalent, $(\parts(p'),q') \notin R_2 \implies (p',q') \notin \MAS$. Hence, the underlined subformula implies $\forall p^* \in P.\,u_1\in U(p^*)\, \forall p' \in \Post^{\Ss}(p^*, u_1).\,(p',q') \notin \MAS.$ Substituting in $\phi,$ it implies $\phi' \coloneqq \exists (P,q) \in R_2.\,\exists u_1\in U_1(P)\,\forall u \in U(q)\,\exists q' \in \Post^\Ss(q)\, \forall p^* \in P.\,u_1\in U(p^*)\, \forall p' \in \Post^{\Ss}(p^*, u_1).\,(p',q') \notin \MAS.$

Now, note that by construction of $R_2$, $(P,q) \in R_2 \implies \exists p \in P.\, (p,q) \in \MAS.$ Using this fact and particularizing $\forall p^*.\,u_1\in U(p^*)$ by $p,$ (which is sound because $u_1 \in U(p)$) $\phi'$ implies $\exists (p,q) \in \MAS\,\exists u_1\in U_1(P)\,\forall u \in U(q)\,\exists q' \in \Post^\Ss(q)\, \forall p' \in \Post^{\Ss}(p, u_1)$ $.\,(p',q') \notin \MAS.$ This is a contradiction to the fact that $\MAS$ is a maximal ASR from $\Ss$ to itself.%
\end{proof}

\begin{proof}[Proof of Lemma \ref{lem:step1mas}]
(1) \textbf{$\MAS_1$ is an AS}: We first show that $\MAS_1$ is indeed an AS from $S_1$ to itself by showing that $\MAS_1$ satisfies all 3 requirements to be an AS from $S_1$ to itself. By definition, $\forall P \in X_1$, $(P,P) \in \MAS_1$. Hence, Requirement 1, trivially holds. 
As \MAS is an AS, $\forall (p,q) \in \MAS. H(p) = H(q)$. Moreover, by definition of $\mathcal{S}_1$, $\forall P \in X_1. \forall p\in P. H(p) = H'(P)$.
By definition (III), $(P,Q) \in \MAS_1$ implies $\forall p \in P,q \in Q. (p,q) \in \MAS_1$, which implies $H'(P)=H'(Q)$ (Requirement 2). 
As $\MAS$ is an AS from $\mathcal{S}$ to itself. Hence, 
$\forall (p,q) \in \MAS. \forall u_1 \in U(p). \exists u_2 \in U(q). \forall (q,u_2, q') \in \delta. \exists (p,u_1, p') \in \delta. (p',q') \in \MAS$.
This along with (I), (II) and (III) implies
$\forall (p,q). (\parts(p),\parts(q))\in \MAS_1 \Rightarrow [\forall u_1 \in U_1(\parts(p)). \exists u_2 \in U_1(\parts(q)) \forall q'. (\parts(q),u_2, \parts(q')) \allowbreak \in \delta \exists p'. (\parts(p),u_1, \parts(p')) \in \delta_1. (\parts(p'),\parts(q')) \in \MAS_1$ which is equivalent to requirement 3.
\\$\MAS_1$ \textbf{is Maximal}: Suppose $\MAS_1$ is not maximal. Hence, $\exists P,Q \in X_1$ such  that $\mathcal{S}_1(P) \leAS \mathcal{S}_1(Q)$ but $(P,Q) \notin \MAS_1$. Note that  $(P,Q) \notin \MAS_1$  implies $\exists p \in P. \exists q \in Q (p,q) \notin \MAS$. However, $\mathcal{S}_1(P) \leAS \mathcal{S}(p)$ and $\mathcal{S}_1(Q)\leAS \mathcal{S}(q)$. By transitivity of $\leAS$, $\mathcal{S}(p) \leAS \mathcal{S}(q)$. This implies, $\MAS$ is not a maximal AS from $\mathcal{S}$ to itself which is a contradiction.
\\(2) As $\MAS_1$ is a maximal alternating simulation relation from $\mathcal{S}_1$ to itself. $\forall P,Q \in X_1 \mathcal{S_1 (P)} \leAS \mathcal{S_1(Q)}$ implies $(P,Q) \in \MAS_1$.$\leAS$ is reflexive and transitive which implies $\MAS_1$ is reflexive and transitive. Suppose $\MAS_1$ is not anti-symmetric. There exists a distinct pair of states $P, Q \in X_1$ such that $(P,Q) \in \MAS_1 \wedge (Q,P) \in \MAS_1$. This, along with the definition of $\MAS_1$ implies, $\exists p \in P. \exists q \in Q. (p,q) \in \MAS \wedge (q,p) \in \MAS$. But this implies, $p$ and $q$ lie in the same partition. Hence, $P=Q$, which is a contradiction.
\end{proof}

\begin{proof}[Proof of Prop.~\ref{prop:c1}]
If $\MAS$ is not anti-symmetric, it has at least one pair of states $p,q$ such that $(p,q)\in \MAS$ and $(q,p)\in \MAS$. While creating the quotient, these states are combined together reducing the number of states by at least 1. Moreover, we never add a new transition in this reduction.
\end{proof}

\begin{proof}[Proof of Lemma \ref{lem:step2correctness}]
To prove (A) the identity map $\mathcal{I}: X \mapsto X$ is an ASR from $\mathcal{S}_2$ to $\mathcal{S}$ and (B) $\MAS$ is an ASR from $\mathcal{S}$ to $\mathcal{S}_2$.

(A) Requirements (i) and (ii) are trivially satisfied as $X_0 = X_{0,2}$, and $H = H_2$. To prove that requirement (iii). Suppose it does not. Then, $\exists q.~\exists u \in U_2(q).~ \forall u' \in U(q). \exists q'' \in \Post^{\Ss}(q,u'). \forall (q' \in \Post^{\Ss_1}(q,u)$  $q' \ne q''$ holds. By definition of $\delta_2$, for all $q \in X$, $U(q) \supseteq U_2(q)$ and $\forall u \in U_2(q).q' \in X$ it holds that $\Post^{\Ss_2}(q,u) = \Post^{\Ss_2}(q,u)$. Just substituting $u'=u$ (in which case we can substitute $\Post^{\Ss_1} = \Post^{\Ss}$) (iii) implies $\Is$ is not an AS from $S_2$ to itself which is a contradiction. 

(B) Requirements (i) and (ii) are trivially satisfied as $X_0 = X_{0,2}$, and $H = H_2$. To show that requirement (iii) is satisfied, we need to show that $\forall (p,q) ~\in \MAS.$  $\forall u_1 \in U(P). ~\exists u_2 \in U_2(q).$ $\forall (q,u_2,q') \in \delta_2. ~\exists (p,u_1, p') \in \delta. ~(p',q')\in \MAS$.
Note that $(q,u_2,q') \in \delta_2$ condition is within the scope of $\exists u_2. u_2 \in U_2(q)$. By construction, $(q,u_2,q') \in \delta_2$ is equivalent to $(q,u_2,q') \in \delta$ when $u_2 \in U_2(q)$. Hence, we need to show $\forall (p,q).~\in \MAS.$ $\forall u_1 \in U(P).$ $\exists u_2 \in U_2(q).~ \forall (q,u_2,q')\in \delta.$ $\exists (p,u_1, p') \in \delta. (p',q')\in \MAS$. In other words, we need to show $\kappa \coloneqq \forall(p,q) \in \MAS.$ $\forall u_1 \in U(P).$ $\exists u_2 \in U_2(q). (p,u_1)\sqsubseteq_{\mathcal{S}} (q,u_2)$.
We know that $\MAS$ is an ASR from $\mathcal{S}$ to itself. Hence, it satisfies $\phi \coloneqq \forall (p,q).~ \in \MAS \forall u_1 \in U(p).~ \exists u_2 \in U(q).~\wedge  (p,u_1)\sqsubseteq{\mathcal{S}} (q,u_2)$. Now observe that
$u_2 \in U(q) \iff u_2 \in U_2(q) \vee u_2\in U(q) \setminus U_2(q)$. By using this identity in the formula $\phi$, we get
$\gamma = \forall (p,q) \in \MAS. \forall u_1 \in U(p). \exists u_2.  [u_2 \in U_2(q) \wedge  (p,u_1)\sqsubseteq_{\mathcal{S}} (q,u_2)] \vee \underline{[u_2\in U(q) \setminus U_2(q) \wedge (p,u_1)\sqsubseteq_{\mathcal{S}} (q,u_2)]}$. We analyze the underlined sub formula. Note that, by construction of $\delta_2$, $u_2 \in U(q) \setminus U_2 (q)$ implies $\exists u_2'\in U_2(q) (q,u_2) \sqsubseteq_{\mathcal{S}} (q,u'_2)$. Using this implication in the underlined formula, we get $\varphi(p,u_1,q,u_2) = [\exists u_2'\in U_2(q).~(q,u_2) \sqsubseteq_{\mathcal{S}} (q,u'_2) \wedge (p,u_1)\sqsubseteq_{\mathcal{S}} (q,u_2)]$. Recall that $\sqsubseteq_{\mathcal{S}}$ is a transitive relation. Hence, $ \varphi \equiv \varphi'(p,u_1,q) = [\exists u_2'\in U_2(q). (p,u_1)\sqsubseteq_{\mathcal{S}} (q,u'_2)]$. Note that $\varphi'$ no longer contains $u_2$ as free variable. Hence, substituting the underlined formula with $\varphi'(p,u_1,q)$ we get:
$\forall (p,q) \in \MAS.~ \forall u_1 \in U(p).~ \exists u_2.~  [u_2 \in U_2(q) \wedge  (p,u_1)\sqsubseteq_{\mathcal{S}} (q,u_2)] \vee \underline{[\exists u_2'. u_2'\in U_2(q) (p,u_1)\sqsubseteq_{\mathcal{S}} (q,u'_2)]}$. Renaming $u_2'$ as $u_2$, we get $\kappa$. Hence, $\phi \implies \kappa$, which proves the result.
\end{proof}

\begin{proof}[Proof of Prop.~\ref{prop:c2}]
If there exists a state $q \in X$ and distinct $u, u' \in U$ such that $(q,u) \preceq{\MAS} (q,u')$ then either $u$ is an irrational action or, both $u$ and $u'$ are equally rational at state $q$. In both the cases the transitions outgoing from $q$ labelled $u'$ will be removed. Hence $|\delta'| < |\delta|$. Note that removing irrational and redundant actions do not introduce new irrational or redundant actions. As $\delta_2$ is constructed from $\delta$ by removing all the transitions corresponding to irrational and redundant actions the former contain only rational moves at every state.
\end{proof}

\begin{proof}[Proof of Lemma \ref{lem:step3correctness}]
Let $Q_{ys}$ and $\delta_{ys}$ be set of all the initial states and transitions, respectively, which are younger siblings. Hence, $X_{0,2} = X_0 \setminus Q_{ys}$ and $\delta_2 = \delta \setminus \delta_{ys}$.  We show (A)Identity map $\mathcal{I}:X \mapsto X$ is an AS from $\mathcal{S}$ to $\mathcal{S}_2$ (B)$\MAS$ is an AS from $\mathcal{S}_2$ to $\mathcal{S}$. 
(A) As $X_{0,2} \subseteq X_0$, Requirement 1 is trivially satisfied. The state set $X$ and the output map $H$ are same in both $\mathcal{S}_2$ and $\mathcal{S}$ implying satisfaction of requirement 2. Suppose requirement 3 is not satisfied. $\exists (p,p) \in \mathcal{I}.~\exists u_1 \in U(p). ~\forall u_2 \in U(p). ~\exists (p,u_2,p') \in \delta_2. ~\forall(p,u_1,p'') \in \delta.~ p' \ne p''$. But $\delta_2 \subseteq \delta$. This leads to a contradiction (as it implies that a transition exists in $\delta_2$ but not in $\delta \supseteq \delta_2$).

(B) We now show that $\MAS$ is an AS from $\mathcal{S}_2$ to $\mathcal{S}$. We need to show that for every initial state $q \in X_0$ there exists an initial state $q' \in X_0$ such that $(q',q) \in \MAS$. As $\MAS$ is reflexive, we have that for every initial state $q \in X_0\setminus Q_{ys}$ there exists an initial state $q \in X_0$ such that $q,q \in \MAS$. For $q \in Q_{ys}$ there exists a $q' \in X_0$ (the ``elder sibling'') such that $(q',q) \in \MAS$. Requirement 2 trivially holds as the set of states and output map is identical in both $\mathcal{S}_2$ and $\mathcal{S}$. 
For proving requirement 3, recall that $\MAS$ is an AS from $\mathcal{S}$ to itself. Hence, $\varphi = \forall (p,q) \in \MAS.~ \forall u_1 \in U(p). ~\exists u_2 \in U(q). ~\forall q'. ~(q,u_2,q') \in \delta \rightarrow~\exists p'. (p,u_1,p') \in \delta \wedge (p',q') \in \MAS$.
As $\delta_2 = \delta \setminus \delta_{ys} $, $\exists p'. ~(p,u_1,p') \in \delta$ is equivalent to $\{\exists p'. ~(p,u_1,p') \in \delta_{ys}\} \vee \{\exists p'. ~(p,u_1,p') \in \delta_2\}$. Thus, we get,
$\varphi' =  \forall (p,q) \in \MAS. ~\forall u_1 \in U(p). ~\exists u_2 \in U(q). ~\forall q'.~ (q,u_2,q') \in \delta \rightarrow [\{\exists p'. ~(p,u_1,p') \in \delta_2 \wedge (p',q') \in \MAS\} \vee   \{\exists p'. ~(p,u_1,p') \in \delta_{ys} \wedge (p',q') \in \MAS)\}]$. $(p,u_1,p') \in \delta_{ys} \wedge (p',q') \in \MAS$ implies(by existence of ``elder brother'' for $p,u_1,p'$) $\exists p''.~ (p,u_1,p'') \in \delta_2 \wedge (p,u_1,p') \in \delta_{ys} \wedge (p'',p') \wedge (p', q')$  implies(by transitivity of $\leAS$) $\{\exists p''. (p,u_1,p'') \in \delta_2 \wedge (p'',q') \}$. Finally substituting this implication in $\varphi'$ we get:
\\$\forall (p,q) \in \MAS.~ \forall u_1 \in U(p). ~\exists u_2 \in U(q). ~ \forall q'. (q,u_2,q') \in \delta. ~$ $[ \{ \exists p'. ~(p,u_1,p') \in \delta_2 \wedge (p',q') \in \MAS\} \vee $  \\$\underline{\{\exists p'. ~\exists p''. ~(p,u_1,p'') \in  \delta_2 \wedge (p'',q')\}}]$

Removing $\exists p'$ from the underlined formula as there is no reference to $p'$ in the formula within this quantifier. 
\begin{quote}
    $\forall (p,q) \in \MAS. \forall u_1 \in U(p). ~\exists u_2 \in U(q).~ \forall q'. ~(q,u_2,q') \in \delta \rightarrow [\{\exists p'. ~(p,u_1,p') \in \delta_2 \wedge (p',q') \in \MAS\} \vee   \underline{\{\exists p''.~ (p,u_1,p'') \in \delta_2 \wedge (p'',q') \in \MAS)\}}]$
\end{quote}

Renaming $p''$ to $p'$ in the underlined subformula and applying idempotence ($\phi \vee \phi \equiv \phi$), we get

\begin{quote}
    $\forall (p,q) \in \MAS.~ \forall u_1 \in U(p).~ \exists u_2 \in U(q).~ \forall q'.~ (q,u_2,q') \in \delta \rightarrow \exists p'.~ (p,u_1,p') \in \delta_2 \wedge (p',q') \in \MAS$
\end{quote}
which is the requirement 3.
\end{proof}

\begin{proof}[Proof of Prop.~\ref{prop:c3}]
If there exists a transition (or an initial state) that is a younger sibling, we remove that transition  from $\mathcal{S}$ (or make the state non-initial) reducing $\trsize(S)$ by at least 1. Note that we do not add or remove any states. Hence, the state size is not affected. By construction, we eliminate all the transitions or initial states which are younger siblings of another transition or initial state, respectively. Note that, removing a younger sibling does not add new younger siblings.
\end{proof}

\begin{proof}[Proof of Lemma \ref{lem:step4correctness}]
 We show that the identity map $\mathcal{I}: X_4 \mapsto X_4$ satisfies all the requirements for being an ASR from $\mathcal{S}$ to $\mathcal{S}_4$ and vice-versa (note that $\mathcal{I} = \{(x,x) \mid x \in X_4\}$ is a subset of both $X \times X_4$ and $X_4 \times X$, hence it is a valid relation in both directions).  Requirement (i) trivially holds (for both directions) as the initial state sets are the same. Requirement (ii) holds (for both directions) as the output map is the same and $\mathcal{I}$ is an identity function. Suppose requirement (iii) does not hold. That is, $\exists (p,p) \in \mathcal{I}. \exists u_1 \in U(p). \forall u_2 \in U_4(p). \exists (p,u_2,p') \in \delta_4. \forall (p,u_1,p'') \in \delta. p' \ne p''$. This is a contradiction since $\delta_4 \subseteq \delta$ (by construction) and the statement implies that there exists a transition in $\delta_4$ not present in $\delta$. For the inverse direction, assume again that requirement (iii) does not hold. That is, 
 $\exists (p,p) \in \mathcal{I}.~ \exists u_1 \in U_4(p).~ \forall u_2 \in U(p).~ \exists (p,u_2,p') \in \delta.~ \forall (p,u_1,p'') \in \delta_4.~ p' \ne p''$. If 
 $p' \in X_4$,  it leads to a contradiction (with $u_2 = u_1$) since $\delta_4$ contains all elements of $\delta$ where the source and the target states are in $X_4$. If $p' \notin X_4$, then by definition $p$ is not reachable from any state in $X_4$ which again is a contradiction as $p \in X_4$.
\end{proof}

\begin{lemma}
\label{simulation}
Let $S, S'$ be any LTS over same set of states $X$. Let $MAS$ be the maximal alternating simulation from $S$ to itself. 
Let $\mathcal{I}:X \mapsto X$ be an identity function. 
$S' \preceq_{\mathcal{I}} S \preceq_{MAS} S'$ (or $S \preceq_{\mathcal{I}} S' \preceq_{MAS} S)$ implies $MAS$ is the maximal alternating simulation relation from $S'$ to itself.  
\end{lemma}
\begin{proof}
Given $S' \preceq_{\mathcal{I}}S \preceq_{MAS} S'$.
\\Hence, $\forall (p,q) \in MAS. S'(P) \leAS S(P) \leAS S(Q)$. This implies (1)$S' \preceq_{\MAS_1} S'$ (by transitive). 
As $MAS$ is also a maximal relation from $S$ to itself, it contains all the reflexive pairs i.e. pairs of the form $(p,p)$.  Hence, $\forall p \in X. S'(p) \leAS S(p) \leAS S' (p)$. Hence,(2) $\forall p \in X. S(p) \equivAS S'(p)$. 
We now show that $MAS$ is the maximal relation satisfying (1). Suppose it is not. Then, there exists $(p',q') \notin \MAS_1$ such that $S'(p') \leAS S'(q')$. By (2) $S(p') \equivAS S'(p') \leAS S'(q') \equivAS S(q')$. This implies $S(p') \leAS S(q')$. But this is a contradiction as $(p',q') \notin MAS$ and $MAS$ is a maximal alternating simulation from $S$ to itself. (3) Hence, $MAS$ is the maximal alternating relation for $S'$ to itself. 

Similar argument as above proves, if $S \preceq_{\mathcal{I}}S' \preceq_{MAS} S$ is satisfied, $MAS$ is a maximal alternating simulation from $S'$ to itself.
\end{proof}

\begin{proof}[Proof of Lemma \ref{lem:necsat}]
Let $T_0 = \mathcal{S}$ and $T_i = S^i(T_{i-1})$. Let $T_i \coloneqq (X_i,X_{0,i}, U_i,Y_i, \delta_i, H_i)$. By construction of $T_2,T_3$. $X_1=X_2=X_3$. Let $X=X_1$. 
Moreover, recall that $X_1 =X_2 = X_3$.

By Proposition~\ref{prop:c1}, we know that $S^1(\mathcal{S})$ satisfies $N_1$ and has a maximal alternating simulation relation $\MAS_1$ to itself which is anti-symmetric. Note that former is equivalent to latter.

By Proposition~\ref{prop:c1} $S^2(S^1(\mathcal{S}))$ satisfies $N_2$ and
By Lemma \ref{lem:step2correctness}, $S^2(S^1(\mathcal{S})) \preceq_{\mathcal{I}}S^1(\mathcal{S}) \prec_{\MAS_1} S^2(S^1(\mathcal{S}))$. By Lemma \ref{simulation}, latter implies $\MAS_1$ is the maximal alternating relation from $S^2(S^1(\Ss))$ to itself. This implies that it satisfies condition $N_1$ too.

By Prop.~\ref{prop:c3}, $T_3$ satisfies $N_3$.  
By Lemma \ref{lem:step3correctness} $S^2(S^1(\Ss))\preceq_{\mathcal{I}}S^3(S^2(S^1(\mathcal{S}))\preceq_{\MAS_1}S^2(S^1(\mathcal{S})$. By Lemma \ref{simulation}, $\MAS_1$ is a maximal alternating simulation from $S^3(S^2(S^1(\mathcal{S}))$ to itself and $\MAS_1$ is already shown to be anti-symmetric. Hence, $S^3(S^2(S^1(\mathcal{S}))$ satisfies $N_1$.
For the sake of readability let $T_2 = S^2(S^1(\mathcal{S}), T_3 = S^3(S^2(S^1(\mathcal{S}))$ We now show that $T_3$ satisfies $N_2$.
We only delete non-deterministic transitions on each action to get $T_3$ from $T_2$ Hence, (1) $\forall p \in X$ $U_2(p)=U_3(p)$. Moreover, to get $T_3$ we only delete transitions which are younger siblings. Hence, (2) $\forall p \in X. \forall a \in U_2(p)$ every state $p'\in \Post_{T_2}(p,a)$ alternately simulates some state $p'' \in \Post_{T_3}(p,a)$ (where $p'' = p'$ if $(p,a,p'')$ was not deleted in $T_3$
else $p''$ is such that $(p,a,p'')$ is an elder sibling to $(p,a,p')$). Suppose $T_3$ has either an irrational action or redundant action at state $p$. This implies (3) there exist distinct $u, u'$ such that $(p,u') \sqsubseteq_{T_3} (p,u)$. In other words, every state $q \in \Post_{T_3}(p,u)$ alternately simulates a state $q'' \in \Post_{T_3}(p,u')$.
Moreover, by construction, (4)$\forall p \in X. \forall b \in U_3(p)$ every state $p'\in \Post_{T_3}(p,b)$ alternately simulates some state $p'' \in \Post_{T_3}(p,b)$. Combining (2),(3), (4) by substituting $a=u \in (2)$ and $b=u'$ in (4), distinct $u,u'$ imply $(p,u') \sqsubseteq_{T_2} (p,u)$. This implies $T_2$ does not satisfy $N_2$ which results in a contradiction.

Finally, note that, deletion of inaccessible states wouldn't affect the equivalence modulo alternating simulation equivalence. Moreover, trivially, removing these states won't add new states equivalent to an existing state, add irrational or redundant moves or, add younger siblings to the transition system hence preserving, $N_1$, $N_2$, $N_3$. And by definition, $\mathcal{S}_{out} = S^4(S^3(S^2(S^1(\mathcal{S})))$ satisfies $N_4$. Hence, $\mathcal{S}_{out}$ is the output of our algorithm satisfying all the above mentioned conditions.
\end{proof}

\section{Other proofs}

\begin{proof}[Proof of Prop.~\ref{prop:etcisthesameasstc}]
Consider the identity relation $I \coloneqq \{(x,x) \mid x \in X\}$ as a trivial ASR from $\Ss_{\mathrm{PETC}}$ to itself. Take any state $x = (q,\tau^\mathrm{low}_q).$ Then $\Post(x,\mathtt{s}) \subset \Post(x,\mathtt{w}),$ where this subset relation is strict. Therefore, the action $\mathtt{w}$ is an irrational action on $x$, thus we can remove it from $x$, preserving ASE by Lemma \ref{lem:step2correctness}. This removal renders $(q,c)$ unreachable for all $c > \tau^\mathrm{low}_q$ (owing also to the fact that, by assumption, any such $(q,c)$ is not initial). Thus, these states are removed, again preserving ASE by Lemma \ref{lem:step4correctness}. The obtained system is as in Def.~\ref{def:reducedtraffic}.
\end{proof}
		
\end{document}